\RequirePackage{amsmath}
\documentclass[final]{amsart}                     
\usepackage{graphicx}
\usepackage{amsthm}
\usepackage{amssymb, amsfonts,verbatim}
\usepackage[title]{appendix}
\usepackage[export]{adjustbox}
\usepackage[singlelinecheck=false,margin=0cm]{caption}
\usepackage{dsfont}
\usepackage{enumitem}
\usepackage{float}
\usepackage{hyperref}
\usepackage{mathtools}
\usepackage{pgfplots}
\usepackage{pgfplotstable}
\pgfplotsset{compat=newest}
\usepackage{subcaption}
\usepackage{tabularx}
\usepackage{textcomp}
\usepackage{tikz}
\usetikzlibrary{datavisualization.formats.functions, arrows.meta,decorations.shapes,decorations.pathreplacing,decorations.markings,arrows,shapes,positioning}
\usepackage{xcolor}
\usepackage{todonotes}
 \newcommand{\mb}{\mathbf}

 \newtheorem{pro}{Proposition}
 \makeatletter
\@namedef{subjclassname@2020}{%
  \textup{2020} Mathematics Subject Classification}
\makeatother

\title{A case study of multiple wave solutions in a reaction-diffusion system using invariant manifolds and global bifurcations
}

\author{Edgardo Villar-Sep\'ulveda\textsuperscript{$\star$}}
\address{\textsuperscript{$\star$}Department of Engineering Mathematics, University of Bristol, Bristol BS8 1TW, United Kingdom.} 
\email{rt20469@bristol.ac.uk}

\author{Pablo Aguirre\textsuperscript{$\star \star$}}
\address{\textsuperscript{$\star \star$}Departamento de Matem\'atica, Universidad T\'ecnica Federico Santa Mar\'ia, Casilla 110-V, Valpara\'iso, Chile.}
\email{pablo.aguirre@usm.cl}  
\author{V\'ictor F. Bre\~na--Medina\textsuperscript{$\dagger$}}
\address{\textsuperscript{$\dagger$}Department of Mathematics, ITAM, R\'io Hondo 1, Ciudad de M\'exico 01080, M\'exico.} 
\email{victor.brena@itam.mx}
\thanks{EVS was partially funded by Programa de Incentivo a la Iniciaci\'on Cient\'ifica PIIC DGIIP-UTFSM and his work was carried out at the Departamento de Matem\'atica, Universidad T\'ecnica Federico Santa Mar\'ia. EVS and PA were partially funded by Proyecto Interno UTFSM PI-LI-19-06. PA also thanks Proyecto Basal CMM Universidad de Chile. VFBM thanks the financial support by Asociaci\'on Mexicana de Cultura A.C.}
\begin{document}



\begin{abstract}
A thorough analysis is performed  to find  traveling waves in a qualitative reaction-diffusion system inspired by  a predator-prey model.
We provide rigorous results coming from a standard local stability analysis, numerical bifurcation analysis, and relevant computations of invariant manifolds to exhibit homoclinic and heteroclinic connections, and periodic orbits in the associated traveling wave system with four components. In so doing, we present and describe a zoo of different traveling wave solutions. In addition, homoclinic chaos is manifested via both saddle-focus and focus-focus bifurcations as well as a Belyakov point. An actual computation of global invariant manifolds near a focus-focus homoclinic bifurcation is also presented to unravel a multiplicity of wave solutions in the model.
\end{abstract}



\maketitle


%
\pagestyle{myheadings}
\thispagestyle{plain}
\markboth{E.~Villar--Sep\'ulveda, P.~Aguirre and V.~Bre\~na--Medina}
{Multiple wave solutions in a reaction-diffusion system}

\section{Introduction}
\label{intro}

Invariant manifold analysis and global bifurcations are among the {\it avant garde} of research topics in nonlinear dynamical systems. 
From the seminal works of  L.~P.~Shil\-ni\-kov \cite{afrai,shil65,shil70} onwards, global bifurcations and invariant manifolds have gained a lot of attention from the dynamical systems community; see, for instance, the survey by Guckenheimer {\it et al.}~\cite{Gu15}, and references therein for different high-profile scenarios including slow-fast dynamics, homoclinic orbits in high-dimensions, and traveling waves phenomena. Geometrically, global bifurcations of vector fields are associated with {invariant manifolds} of equilibria and/or periodic orbits, quasiperiodic invariant tori~\cite{guckenheimer,sandstede,kuznetsov}, and slow manifolds of systems with multiple timescales~\cite{jose1,jose2}. Perturbations of the system parameters typically cause global rearrangements of such invariant manifolds, leading to reorganizations of the overall dynamics in the phase space. These can bring forth the creation/destruction of homoclinic and heteroclinic connections, and the (dis)appeareance of attracting (repelling) invariant objects, including periodic orbits and even chaotic dynamics~\cite{kuznetsov,shilshil}. Hence, global invariant manifolds emerge as the ``building blocks" of a dynamical system, i.e., those essential objects help to explain how the overall ``architecture" of phase space is organized.

One of the applications of invariant manifold analysis lies on the existence of traveling wave solutions for reaction-diffusion equations. 
This kind of  solution emerges as a suitable mathematical approach to describe wave-like
spatial movement of populations, transport of nutrients and biological substances, etc.; see, for instance, the textbooks~\cite{mathbio2,mathbio1} and references therein. 
Indeed, traveling waves represent spatiotemporal transitions from one homogeneous steady state to another one, or to itself~\cite{sandstede,Li,mathbio1,tw,wu14}. 
Typically, the mathematical analysis to find this kind of solution involves the reduction of the reaction-diffusion equations into a system of ordinary differential equations in which one searches for heteroclinic or homoclinic orbits.
However, the problem of obtaining these connecting orbits and  associated global (un)stable manifolds is a challenging task. 
With the exception of a few concrete examples (see, e.g.,~\cite{sandstede}), in general, there are no analytical expressions for homoclinic orbits or non-local normal forms. Hence, it is frequent to make use of reductions to Poincar\'e maps in suitable cross sections, center manifold reductions and other analytical approaches to prove the existence of intersecting invariant manifolds giving rise to homoclinic and heteroclinic connections; see~\cite{twpredatorprey,hsu,twpredatorprey3,manna,sandstede98,Takahashi2005,Yamashita2018} for different examples.


The purpose of this paper is to establish the existence of traveling wave solutions for the following reaction-diffusion system:
\begin{eqnarray}
	\left\{ \begin{array}{l}
	u_t=D_1u_{xx}+u(u-m)(1-u)(u+v)-\alpha uv\,,
	\\[1.5ex]
	v_t=D_2v_{xx}+\beta uv-\gamma v(u+v)\,,
	\end{array}
	\right. \label{difusion}
\end{eqnarray}
where  short notation is used for partial derivatives: $w_t=\partial w/\partial t$ and $w_{xx}=\partial^2 w/\partial x^2$. 
In~\eqref{difusion}, $u=u(x,t)$ and $v=v(x,t)$ are the unknown variables as functions of the spatial variable $x\in [-L,L]$ and time $t>0$. The characteristic length of the state variables interaction domain is assumed to be such that $L\to\infty$ as traveling waves solutions are known to arise in systems as~\eqref{difusion} (see e.g.~\cite{VBrena,twpredatorprey3,lewisli,dolnik});  diffusion rates $D_1,D_2>0$ correspond to mobilities which are a measure of the spatial dispersion efficiency of $u$ and $v$, respectively~\cite{mathbio1}.

System~\eqref{difusion} is inspired by a predator-prey model from~\cite{aguirre} accounting for strong Allee effect on prey and ratio-dependent functional response. However, the main purpose of~\eqref{difusion} is not to duplicate exactly quantitative aspects of the predator-prey interactions from the model it is inspired by. Rather, \eqref{difusion} is meant as an elementary, minimal model in which one can display the sorts of mathematical relations between variables underlying the connections in~\cite{aguirre}. Hence, we refrain from calling $u$ and $v$ as the prey and predators, respectively, in order to avoid misunderstandings with the interpretation of results of the conceptual model~\eqref{difusion}.
We note that our approach to \eqref{difusion} is similar to that of other qualitative models in biology such as the celebrated FitzHugh-Nagumo model for action potentials in neurons~\cite{mathbio1}, or the Izhikevich~\cite{izhi}, Hindmarsh-Rose~\cite{a_shil} and the canonical Ermentrout-Kopell~\cite{thetamodel} models. Indeed, while these abstract models are constructed less closely to physiological features and thus less interpretable, they succeed in portray diverse essential neuronal behaviors with just the minimal mathematical ingredients~(see~\cite{mathfoundneuro} and references therein).

With a strategic combination of numerical methods for invariant manifolds and bifurcation theory, we find the traveling wave solutions identifying each of them as a specific heteroclinic/homoclinic connection or a periodic orbit in the four-dimensional phase space of the associated ODEs. We classify these solutions into 12 different classes depending on the topological type of the associated orbit. We also determine conditions on the model parameters so that there is such a particular kind of solution and identify homoclinic chaotic dynamics as one of the sources of complicated behavior. 

Today, homoclinic and heteroclinic orbits can readily be computed and continued with software packages like {\sc Auto}~\cite{auto} (with its extension {\sc HomCont}~\cite{homcont}) and {\sc Matcont}~\cite{Dhooge} with high accuracy. In addition, one can locate homoclinic and heteroclinic connections as intersections of global invariant manifolds. This can be achieved by direct computation and inspection of the manifolds~\cite{ako13,shilnikov}, or by setting additional techniques such as Lin's method~\cite{bk-tr}. While one-dimensional invariant manifolds can be approximated using straightforward integration from a given initial condition, the computation of higher-dimensional manifolds of equilibria and periodic orbits requires advanced numerical techniques. 
The two-dimensional global manifolds in this paper are computed as families of orbit segments, which can be obtained as solutions of a suitable boundary value problem (BVP), irrespective of the vector field undergoing a homoclinic or heteroclinic bifurcation. This allows us to make use of {\sc Auto} to implement and solve the BVP; then, the manifold is ``built up'' by continuation of the respective orbit segments~\cite{numericalmanifold,krauskopf}; see also~\cite{aguirrehom,adko11,shilnikov,Gu15} for further details and applications. 
Moreover, while some works have dealt with traveling waves associated with three-dimensional vector fields~\cite{lin,Takahashi2005,Yamashita2018}, our problem involves a four-dimensional phase space, which is a major challenge. Although the human brain is efficient when capturing depth in flat images of 3-dimensional objects, this ability is not as effective in higher dimensions~\cite{siads20,Osinga}. When trying to visualize objects in a 4D phase space ---such as invariant manifolds---, standard projections may give rise to false intersections between them. These artifacts due to projections must be detected and differentiated from real intersections to ensure or discard the existence of homoclinic or heteroclinic connections. To do so, we make extensive use of dynamical systems theory and topological arguments to justify our findings. In addition, wave trains are found as periodic orbits originating via Hopf bifurcations; see also~\cite{twpredatorprey,hsu,twpredatorprey3,manna,claire} for other uses of this theoretical argument.

This paper is organized as follows: \S\ref{sec:preliminares} presents some definitions, notation and preliminary results. Local stability analysis of steady states is included in \S\ref{sec:local}, while a bifurcation analysis is presented in \S\ref{sec:bif}. Wave pulses, wave trains and wave fronts are discussed in \S\ref{sec:homoclinic}, \S\ref{sec:cycles} and \S\ref{sec:heteroclinic}, respectively. \S\ref{sec:focusfocus} presents a description of the multiple wave fronts found near a focus-focus homoclinic bifurcation. \S\ref{sec:cd-plane} discusses the influence of the propagation speed and the diffusion ratio on the occurrence of the different wave pulses. \S\ref{sec:planes} analyzes the existence of traveling fronts in two invariant planes. Finally, \S\ref{sec:discussion} presents a summary and discussion of the main results.

{\color{green}

}

\section{Preliminaries and first examples}
\label{sec:preliminares}

For future convenience, the first step to study traveling wave solutions in \eqref{difusion} is to make a time rescaling and a change of parameters given, respectively, by
$$
	t\to D_2t, \qquad (d,s,b,g,a,m)=\left(\dfrac{D_1}{D_2},\dfrac{1}{D_2},s\beta,s\gamma, s\alpha,m\right)\in  \mathbb{R}^5_+\times ]0,1[. \label{param}
$$
Thus, we can write system \eqref{difusion} equivalently as
\begin{gather}
	\left\{\begin{array}{l}
	u_t=du_{xx}+su(u-m)(1-u)(u+v)-auv;
	\\[1.5ex]
	v_t=v_{xx}+buv-gv(u+v).
	\end{array}
	\right. \label{edp}
\end{gather}
Here, $d=D_1/D_2$ represents the ratio of diffusion rates of $u$ and $v$, respectively, and appears as an explicit parameter in \eqref{edp}. If $d>1$ (resp. $d<1$),  $u$ is more (resp. less) efficient to disperse in space compared to $v$.

We now consider the so-called wave variable $z= x+ct$, where $c>0$ is the wave speed, and we look for solutions of \eqref{edp} of the form $U(z)=u(x,t)$, $V(z)=v(x,t)$. Applying the chain rule and substituting this into \eqref{edp}, we obtain the following set of second order ordinary differential equations 
\begin{eqnarray}
	\left\{\begin{array}{l}
	c\dfrac{dU}{dz}=d\dfrac{d^2U}{dz^2}+sU(U-m)(1-U)(U+V)-aUV;
	\\[1.5ex]
	c\dfrac{dV}{dz}=\dfrac{d^2V}{dz^2}+bUV-gV(U+V).
	\end{array}
	\right. \label{edp2}
\end{eqnarray}

Naming the auxiliary variables $W=dU/dz$ and $R=dV/dz$,
system \eqref{edp2} can be expressed equivalently as the vector
field
\begin{eqnarray}
	X:\left\{\begin{array}{l}
	\dfrac{dU}{dz}=W,
	\\[1.5ex]
	\dfrac{dV}{dz}=R,
	\\[1.5ex]
	\dfrac{dW}{dz}=\dfrac{1}{d}\left(cW-sU(U-m)(1-U)(U+V)+aUV\right),
	\\[1.5ex]
	\dfrac{dR}{dz}=cR-bUV+gV(U+V).
	\end{array}
	\right. \label{sist}
\end{eqnarray}

 Inspired by the biological origins of  \eqref{edp},  we restrict our analysis of \eqref{sist}  to 
the set $\Omega=\{(U,V,W,R)\in\mathbb{R}^4:\,\,U\geq0,V\geq0\}$.
A traveling wave of \eqref{edp} is any bounded solution of the system \eqref{sist} contained in the domain $\Omega$. These solutions are functions that ``travel" in space, preserving their form as time goes by~\cite{sandstede,mathbio2,mathbio1}. 
%

As we need to find bounded solutions of the system \eqref{sist}, we focus our attention on three special types of traveling waves: wave pulses, wave fronts, and wave trains~\cite{sandstede}. 
%
%
%
 Fig.~\ref{fig:pulsodeonda}(a) shows an actual homoclinic orbit of \eqref{sist} found with the method presented later, in \S\ref{sec:homoclinic}. The homoclinic orbit connects the equilibrium point $\mb{q}=(q_u,q_v,0,0)$ (given explicitly in \S\ref{sec:local}) to itself.  The homoclinic connection is an orbit in the unstable manifold of $\mb{q}$, $W^u(\mb{q})$, that comes back to $\mb{q}$ along its stable manifold $W^s(\mb{q})$, i.e., it is in the intersection $W^u(\mb{q})\cap W^s(\mb{q})$.
 The time series of $U$ and $V$ associated with this connecting orbit
are shown in Fig.~\ref{fig:pulsodeonda}(b) in blue and orange, respectively. The profile of this solution is characterized by a large deviation (or pulse) in the amplitudes of $U$ and $V$ followed by a convergence back to the resting state. This corresponds to a wave pulse in the original system \eqref{edp} that travels from a spatially homogeneous stationary solution to itself. Therefore, homoclinic orbits of~\eqref{sist} correspond to wave pulses of \eqref{edp}. 
%

\begin{figure}
	\centering
	\includegraphics[width=\textwidth]{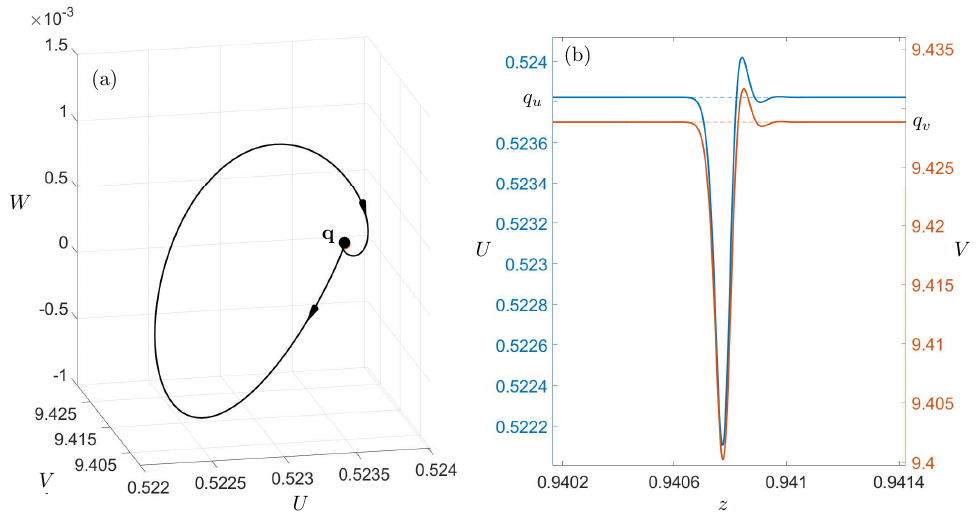}
	\caption{Profile of a wave pulse. Panel (a) shows a homoclinic orbit which joins the equilibrium $\mb{q}$ to itself in the long term, projected onto the $UVW$ space, while panel (b) shows the time series of $U$ and $V$ associated with this solution in blue and orange, respectively; here, the range of $z$ values in the horizontal axis is restricted to the interval where the variables $U$ and $V$ develop the pulse. Parameter values are $a=24$, $b=19$, $g=1$, $m=0.0463358$ are taken from \cite{aguirre} while $c=1$, $s=100$, $d=1.3080156$ are chosen after a bifurcation analysis in \S\ref{sec:bif}.}
	\label{fig:pulsodeonda}
\end{figure}

Fig.~\ref{fig:frentedeonda}(a) shows a heteroclinic orbit connecting $\mb{q}$ to the equilibrium $\mb{p}=(p_u,p_v,0,0)$ (given explicitly below, in \S\ref{sec:local}).  The heteroclinic connection is an orbit in $W^u(\mb{q})$ which moves away from $\mb{q}$, but it is also contained in the stable manifold of $\mb{p}$ $W^s(\mb{p})$, and hence, heads toward $\mb{p}$. That is, this heteroclinic orbit lies in $W^u(\mb{q})\cap W^s(\mb{p})$. This corresponds to a wave of \eqref{edp} that makes the transition from one spatially homogeneous stationary solution to another as is shown in Fig.~\ref{fig:frentedeonda}(b). 
%

\begin{figure}
	\centering
	\includegraphics[width=\textwidth]{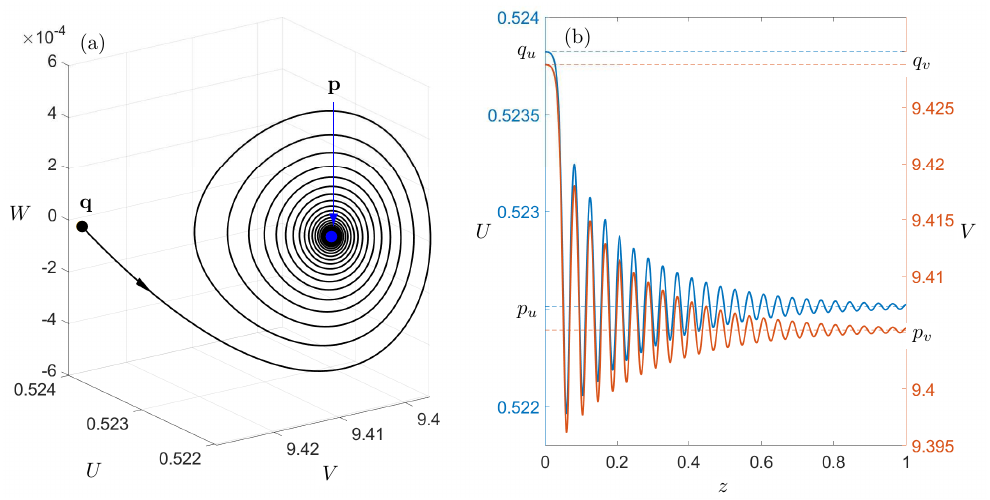}
	\caption{Profile of a wave front. Panel (a) shows a heteroclinic orbit from $\mb{q}$ to $\mb{p}$, projected onto the $UVW$ space, while panel (b) shows the time series of $U$ and $V$ associated with the same solution. Parameter values are the same as in Fig.~\ref{fig:pulsodeonda} except for $d=2.2883206$, and $c=0.4372925$.}
	\label{fig:frentedeonda}
\end{figure}

The homoclinic and heteroclinic orbits (and their associated time series) as solutions of \eqref{sist} are parameterized by $z\in(-\infty,\infty)$.  However, for computational purposes, this independent variable is rescaled to $z\in[0,1]$ in all our results; this is actually a standard procedure with numerical continuation methods (see~\cite{doe-book,numericalmanifold} and references therein). Moreover, in Fig.~\ref{fig:pulsodeonda} (and for every other wave pulse shown throughout this paper) we restrict the values of $z$ to those compact subintervals of $]0,1[$ where the pulses are easier to see.


The third kind of wave solution of \eqref{edp} we are interested in is wave trains. These solutions correspond to periodic orbits of system \eqref{sist}, as is illustrated in Fig.~\ref{fig:trendeonda}. 
%
For computational purposes, the period of every periodic orbit of \eqref{sist} is rescaled to $T=1$; see~\cite{doe-book,numericalmanifold}. In particular, in Fig.~\ref{fig:trendeonda} (and for every other wave train shown throughout this paper) we restrict the values of $z$ to one such period of the cycle.

\begin{figure}
	\centering
	\includegraphics[width=\textwidth]{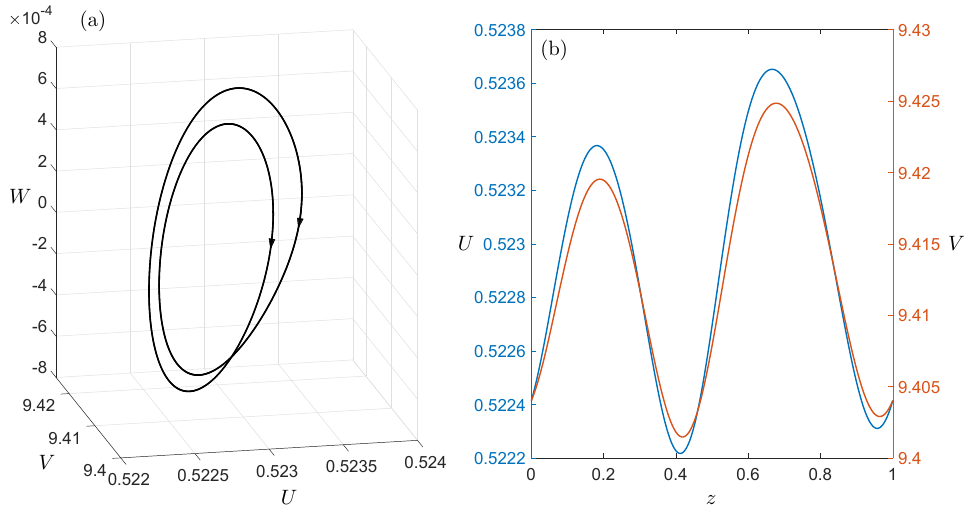}
	\caption{Profile of a wave train. Panel (a) shows a periodic orbit projected onto the $UVW$ space, while panel (b) shows the time series of $U$ and $V$ contained in a (rescaled) period of length 1. Parameter values are the same as in Fig.~\ref{fig:pulsodeonda} except for $d= 1.3107$. }
	\label{fig:trendeonda}
\end{figure}

\section{Local stability analysis}
\label{sec:local}

System \eqref{sist} has at most five equilibrium points in $\Omega$, which are given by $\mb{p}_0=(0,0,0,0)$, $\mb{p}_m=(m,0,0,0)$, $\mb{p}_1=(1,0,0,0)$, $\mb{p}=(p_u,p_v,0,0)$, and $\mb{q}=(q_u,q_v,0,0)$,
where
\\
\begin{subequations}\label{eq:quvpuv}
\begin{minipage}{0.45\textwidth}
	\begin{align}
	q_u&=\dfrac{bs(1+m)+\sqrt{bs\Delta}}{2bs},\label{eq:qu}
	\\[1.5ex]
	q_v&=\dfrac{(b-g)}{g}q_u,\label{eq:qv}
	\end{align}
\end{minipage}
\hfill
\begin{minipage}{0.45\textwidth}
	\begin{align}
	p_u&=\dfrac{bs(1+m)-\sqrt{bs\Delta}}{2bs},\label{eq:pu}
	\\[1.5ex]
	p_v&=\dfrac{(b-g)}{g}p_u,\label{eq:pv}
	\end{align}
\end{minipage}
\end{subequations}
\\[1.5ex]
provided that $\Delta=bs(m-1)^2-4a(b-g)\geq0$, and $q_u,q_v,p_u,p_v\geq 0$.
Under these conditions, we have the following result on the stability of $\mb{p}_0$, $\mb{p}_m$ and~$\mb{p}_1$:

\begin{pro} \label{lema1}
	Let us consider the quantities
	$$
		\Delta_m^1:=c^2+4(g-b)m, \qquad \Delta_m^2:=c^2-4d(1-m)m^2s.
		$$
	Then, system \eqref{sist} satisfies the following statements:
	\begin{enumerate}
		\item $\mb{p}_0$ is an unstable non-hyperbolic equilibrium, $\dim(W^u(\mb{p}_0))=2$, and $\dim(W^c(\mb{p}_0))=2$. \label{uno}
		\item If $b<g$, then $\mb{p}_m$ and $\mb{p}_1$ are hyperbolic saddles,  $\dim(W^s(\mb{p}_m))=1$, $\dim(W^u(\mb{p}_m))=3$, $\dim(W^s(\mb{p}_1))=2$, and $\dim(W^u(\mb{p}_1))=2$. \label{dos}
		\item If $b>g$, then $\mb{p}_m$ is a hyperbolic repeller, and $\mb{p}_1$ is a hyperbolic saddle, $\dim(W^s(\mb{p}_1))=1$, and $\dim(W^u(\mb{p}_1))=3$. In addition, if $\Delta_m^1>0$ and $\Delta_m^2>0$, then $\mb{p}_m$ is a repelling node. \label{tres}
	\end{enumerate}
\end{pro}

\begin{proof} If we denote the vector field \eqref{sist} by $X$,
	its Jacobian matrix evaluated at the points $\mb{p}_0$, $\mb{p}_m$ and $\mb{p}_1$ is given, respectively, by:
	$$
		DX(\mb{p}_0)=\begin{pmatrix}
		0 && 0 && 1 && 0
		\\
		0 && 0 && 0 && 1
		\\
		0 && 0 && \dfrac{c}{d} && 0
		\\
		0 && 0 && 0 && c
		\end{pmatrix},
		\hspace{5mm}
		DX(\mb{p}_m)=\begin{pmatrix}
		0 && 0 && 1 && 0
		\\
		0 && 0 && 0 && 1
		\\
		-\dfrac{(1-m)m^2s}{d} && \dfrac{am}{d} && \dfrac{c}{d} && 0
		\\
		0 && (g-b)m && 0 && c
		\end{pmatrix},
		$$
		$${\rm and}\  \
		DX(\mb{p}_1)=\begin{pmatrix}
		0 && 0 && 1 && 0
		\\
		0 && 0 && 0 && 1
		\\
		\dfrac{(1-m)s}{d} && \dfrac{a}{d} && \dfrac{c}{d} && 0
		\\
		0 && g-b && 0 && c
		\end{pmatrix}.
	$$
	Denoting by $\lambda_i^j$ the $i$-th eigenvalue of the equilibrium $\mb{p}_j$, $j\in\{0,1,m\}$, then we see that
	\begin{align*}
		\begin{cases}
		\lambda_{1,2}^0&=0,
		\\[1.5ex]
		\lambda_3^0&=\dfrac{c}{d}>0,
		\\[1.5ex]
		\lambda_4^0&=c>0.
		\end{cases} \quad
		\begin{cases}
		\lambda_{1,2}^m&=\dfrac{c\pm \sqrt{\Delta_m^1}}{2},
		\\[1.5ex]
		\lambda_{3,4}^m&=\dfrac{c\pm \sqrt{\Delta_m^2}}{2d},
		\end{cases} \quad
		\begin{cases}
		\lambda_{1,2}^1&=\dfrac{c\pm \sqrt{c^2+4(g-b)}}{2},
		\\[1.5ex]
		\lambda_{3,4}^1&=\dfrac{c\pm \sqrt{c^2+4d(1-m)s}}{2d}.
		\end{cases}
	\end{align*}
	Since $DX(\mb{p}_0)$ has two zero and two positive eigenvalues, then $\mb{p}_0$ is an unstable non-hyperbolic equilibrium. The remaining statements are direct consequences of the Hartman-Grobman theorem and the stable manifold theorem~\cite{guckenheimer}. As for $\mb{p}_m$, since $0<m<1$, then $\Delta_m^2=c^2-4d(1-m)m^2s<c^2$. Furthermore, if $b<g$, then $\Delta_m^1=c^2+4(g-b)m>c^2$. Thus we have that
$$\lambda_1^m=\dfrac{c+\sqrt{\Delta_m^1}}{2}>0,\,\quad \lambda_2^m=\dfrac{c-\sqrt{\Delta_m^1}}{2}<0,\,\quad \text{Re}(\lambda_{3,4}^m)>0,$$
	which implies the desired result. On the other hand, if $b>g$, then $\Delta_m^1<c^2$, which implies that $\lambda_2^m>0$. In the particular case that $b>g$, $\Delta_m^1>0$ and $\Delta_m^2>0$, the eigenvalues $\lambda_{3,4}^m$ of $DX(\mb{p}_m)$ become real and positive and, hence, $\mb{p}_m$ is a repelling node. Similarly, a sign analysis of the eigenvalues of $DX(\mb{p}_1)$ reveals the stability of $\mb{p}_1$ and the dimensions of its invariant manifolds. $\qed$
\end{proof}

\begin{pro}
    If $b>g$ and $\Delta>0$, then both equilibria $\mb{p}$ and $\mb{q}$ of \eqref{sist} are in the domain~$\Omega$.
\end{pro}

\begin{proof}
    	It is immediate to see that if $\Delta>0$, then $\mb{p}$ and $\mb{q}$ exist and are different. To see that $\mb{p},\mb{q}\in \Omega$, note that if $b>g$ and $\Delta>0$, then
	\begin{align*}
		bsm &>-a(b-g)
		\\[1.5ex]
		\Leftrightarrow 2bsm&> -4a(b-g)-2bsm
		\\[1.5ex]
		\Leftrightarrow bs(m^2+2m+1)&> -4a(b-g)+bs(m^2-2m+1)
		\\[1.5ex]
		\Leftrightarrow b^2s^2(1+m)^2&> bs\Delta
		\\[1.5ex]
		\Leftrightarrow bs(1+m)&> \sqrt{bs\Delta}
		\\[1.5ex]
		\Leftrightarrow q_u>p_u&> 0.
	\end{align*}
	Finally, since $q_v=(b-g)q_u/g$ and $p_v=(b-g)p_v/g$, the result follows.~$\qed$
\end{proof}

To perform a stability analysis of the equilibria $\mb{p}$ and $\mb{q}$ by standard methods is a challenging task. Indeed, the Jacobian matrix of $X$ evaluated at $\mb{p}$ and $\mb{q}$ are given, respectively, by
\begin{align*}
	DX(\mb{p})=\begin{pmatrix}
		0 && 0 && 1 && 0
		\\
		0 && 0 && 0 && 1
		\\
		a_{31} && a_{32} && \dfrac{c}{d} && 0
		\\
		a_{41} && a_{42} && 0 && c
	\end{pmatrix}
\quad {\rm and} \quad 
    DX(\mb{q})=\begin{pmatrix}
    0 && 0 && 1 && 0
    \\
    0 && 0 && 0 && 1
    \\
    b_{31} && b_{32} && \dfrac{c}{d} && 0
    \\
    b_{41} && b_{42} && 0 && c
    \end{pmatrix},
\end{align*}
where
\begin{align*}
	a_{31}&=\dfrac{2 a \left(b g+g^2-2 b^2\right)+b \left(b (m-1)^2 s-(m+1) \sqrt{b s \Delta}\right)}{2 b d g}p_u,\\
	b_{31}&=\dfrac{2 a \left(b g+g^2-2 b^2\right)+b \left(b (m-1)^2 s+(m+1) \sqrt{b s \Delta}\right)}{2 b d g}q_u,
\end{align*}	
	and
\begin{gather*}
	a_{32}=\dfrac{a g}{b d}p_u, \, \quad
	a_{41}=-(b-g)p_v, \,
	a_{42}=(b-g)p_u.\\
    b_{32}=\dfrac{a g}{b d}q_u,\, \quad
    b_{41}=-(b-g)q_v,\,
    b_{42}=(b-g)q_u.
\end{gather*}

Hence, the computation of analytic expressions for the eigenvalues of $DX(\mb{p})$ and $DX(\mb{q})$ turns out to be a cumbersome goal. However, direct inspection of equilibrium coordinates reveals evidence of some local bifurcations. If $b=g$, then from \eqref{eq:quvpuv} we have that $q_v=p_v=0$, $q_u=1$, and $p_u=m$, since
\begin{gather*}
		\dfrac{bs(1+m)\pm \sqrt{b^2s^2(m-1)^2}}{2bs}
	=\dfrac{1+m\pm  (1-m)}{2}.
\end{gather*}
	This implies that $\mb{q}=\mb{p}_1$ and $\mb{p}=\mb{p}_m$. 	Since, according to Lemma \ref{lema1}, there is a stability change for $\mb{p}_1$ and $\mb{p}_m$ when $b=g$, this is an indication of a transcritical bifurcation. The equilibrium points $\mb{q}$ and $\mb{p}_1$ collide and interchange their stability when $b=g$; a similar statement follows for $\mb{p}$ and $\mb{p}_m$. 
	On the other hand, 
	if $\Delta=0$, then $q_u=p_u$ and $q_v=p_v$, so that $\mb{q}=\mb{p}$. Since these two equilibria exist only if $\Delta>0$, this is evidence of a saddle-node (or fold) bifurcation of equilibrium points. Formal proofs for these statements require the reduction of \eqref{sist} into a one-parameter family of center manifolds on each case, followed by verification of certain genericity conditions; we refer to~\cite{guckenheimer,kuznetsov} for more details. However, we opt to omit these proofs in favor of a focus on the analysis of global bifurcations in \eqref{sist} and the emergence of traveling waves in \eqref{edp}.

\section{Bifurcation analysis}
\label{sec:bif}

In this section, we present a bifurcation analysis of \eqref{sist} performed with the standard continuation
package {\sc Auto}. As a starting point, we consider parameters $a=24$, $b=19$, $g=1$ and $s=100$ fixed throughout this section and let $d$ and $m$ to vary. The fixed values of $a,b,g$ and $s$ correspond to those in \cite{aguirre} 
after the transformation \eqref{param}.  As in general the wave speed is a continuous function of the  parameters system, i.e., $c=c(d,s,b,g,a,m)$, for the purpose of simplification we take the initial value $c = 1$ for the wave speed in \eqref{sist}. This approach can be thought of as an exploratory phase in which one navigates the possible preimages of $c(d,s,b,g,a,m)=1$ in parameter space that allow solutions in the traveling frame of reference moving at speed $c=1$. Later in \S\ref{sec:cd-plane}, we let $c$ to vary in order to capture the existence and properties of the wave solutions by means of a wider range of preimage values of $c$.


\begin{figure}
	\centering
	\includegraphics[trim = 10mm 0mm 10mm 0mm, clip,scale=1]{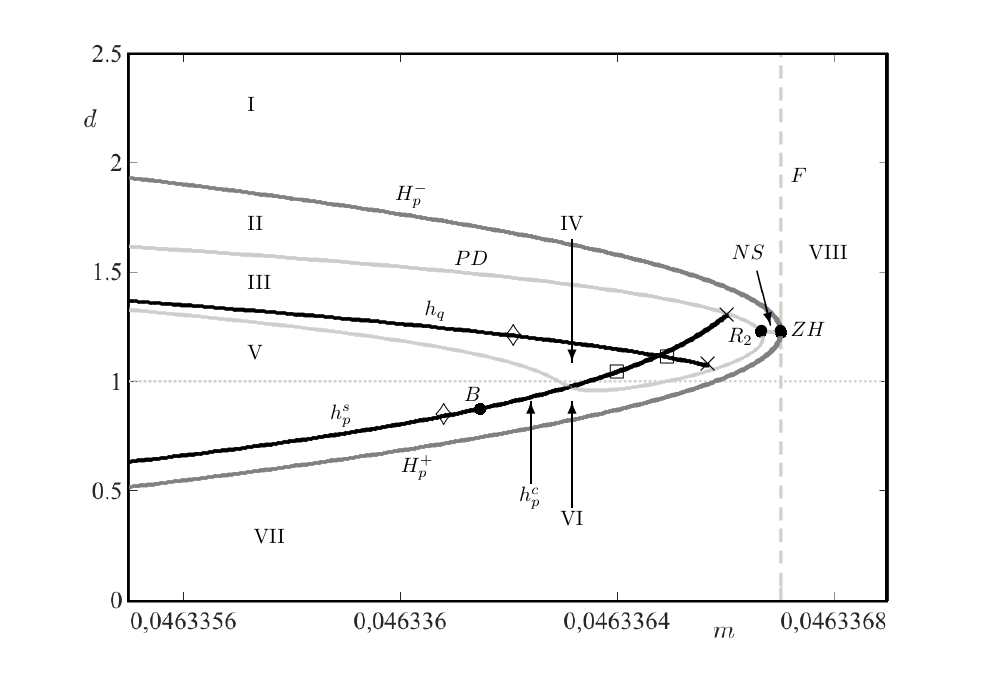}
	\caption{Bifurcation diagram of \eqref{sist} in the $(m.d)$-plane space. 
	Parameter values are the same as in Fig.~\ref{fig:pulsodeonda}.}
	\label{fig:bifurcacion}
\end{figure}

The resulting bifurcation scenario in the $(m,d)$-plane is shown in Fig.~\ref{fig:bifurcacion}. 
Of special importance for us are the curves $h_p$ and $h_q$ which represent homoclinic bifurcations to $\mb{p}$ and $\mb{q}$, respectively.  We will address the technical details and consequences of these homoclinic bifurcations later in \S\ref{sec:homoclinic}. For the moment, it suffices to say that $h_p$ is divided into two segments (labeled as $h_p^c$ and $h_p^s$, respectively) by a codimension two Belyakov homoclinic point, labeled as $B$~\cite{shilnikov,bel84}. 
The {right-hand} side endpoints of both $h_p^c$ and $h_q$ (marked with $\times$) correspond to the last points where we could obtain convergence of the computed solutions with {\sc Auto}. 

 Fig.~\ref{fig:bifurcacion} also shows a curve of Hopf bifurcation at the equilibrium $\mb{p}$. This bifurcation curve is divided into two segments. The first one is a segment of supercritical Hopf bifurcation, labeled as $H_p^-$; the other one is a segment of subcritical Hopf bifurcation, labeled as $H_p^+$.  The separation between both sides of the Hopf curve occurs at a codimension two Zero-Hopf bifurcation point (labeled as $ZH$) where the Hopf curve meets a Fold bifurcation curve $F$ in a quadratic tangency.
 The curve labeled as $PD$ corresponds to a period doubling bifurcation, while $NS$ is a Neimark-Sacker (or torus) bifurcation curve. These two curves meet at a codimension two strong resonant point, $R_2$. The horizontal dotted line in Fig.~\ref{fig:bifurcacion} corresponds to $d=1$ (or equivalently, $D_1=D_2$ in \eqref{difusion}). While this line does not represent any bifurcation, it is useful to distinguish the phenomena encountered above it from that which occurs below it. Indeed, one must remember that if $d>1$ (resp. $d<1$), then $\mb{u}$ has a higher (resp. lower) diffusion rate than $\mb{v}$.

The bifurcation curves in Fig.~\ref{fig:bifurcacion} divide the $(m,d)$-plane into the open regions I-VIII. Region VIII is bounded to the left by the curve $F$, while I is delimited to its right by $F$, and below by $H_p^-$. Region II is surrounded by the curves $H_p^-$, $NS$, and $PD$; while region III is enclosed by the curves $PD$, $h_p^+$, and $h_q$. Furthermore, region IV is bounded by the segments $h_q$, $h_p^+$, and $PD$, while region V is surrounded by the curves $PD$, $h_p^-$, and $h_p^+$. Finally, region VI is enclosed by the curves $h_p^-$, $h_p^+$, $PD$, $NS$, and $H_p^+$, while region VII is delimited to the right by the curve $F$ and above by $H_p^+$.

It is relevant to note that the fold curve $F$ corresponds to the equation $\Delta=0$ in \S\ref{sec:local}. Hence, equilibrium points $\mb{p}$ and $\mb{q}$ exist on the {left-hand} side of the $F$ curve (regions I-VII). In particular, if $(m,d)\in{\rm I}$, both equilibria are hyperbolic. If $(m,d)$ passes through the supercritical Hopf bifurcation curve $H_p^-$ from region I towards region II, a limit cycle 
branches out from $\mb{p}$. While the Hopf bifurcation is supercritical, this criticality is restricted only to a suitable two-dimensional center manifold where the bifurcation takes place~\cite{guckenheimer,kuznetsov}; 
the resulting periodic orbit in $\mathbb{R}^4$ is, in fact, of saddle type in region II. The stability properties of this cycle remain unchanged until this orbit undergoes a period doubling bifurcation when $(m,d)\in PD$. This periodic orbit faces a number of further bifurcations (not shown in Fig.~\ref{fig:bifurcacion}) as the point $(m,d)$ moves towards the curve $h_q$ where it gives rise to a homoclinic orbit. We will address this transition again in \S\ref{sec:cycles}. On the other hand, when the point $(m,d)$ crosses the $NS$ curve from region II into region VI, an invariant torus bifurcates as a periodic orbit undergoes a Neimark-Sacker bifurcation.

Our bifurcation diagram in Fig.~\ref{fig:bifurcacion} is just partially complete. Other codimension two strong resonances can be found along the $NS$ bifurcation curve. Although the complete bifurcation diagram near these bifurcation points is yet to be known in its full complexity, one should expect the appearance of chaotic behavior when the point $(m,d)$ is in a neighborhood of the $NS$ curve; for further details, see \cite{kuznetsov}. 
Furthermore, bifurcation theory tells us that there is an infinite number of bifurcation curves in neighborhoods of both points $B$ and $R_2$. However, the full bifurcation picture near each of these points is {not fully known from a theoretical point of view~\cite{kuznetsov}. (We will  address the complex dynamics that emerges due to the point $B$ in \S\ref{sec:homoclinic} below).

\section{Homoclinic bifurcations, wave pulses, and chaos}
\label{sec:homoclinic}

\begin{figure}
	\centering
	\includegraphics[width=\textwidth]{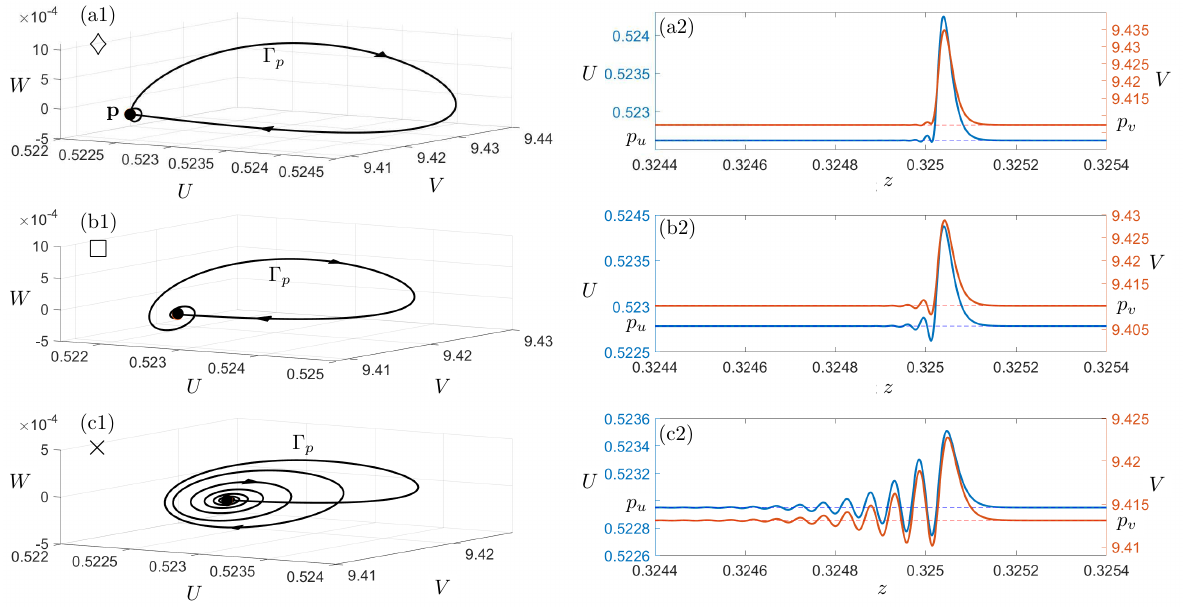}
	\caption{Homoclinic orbits to the equilibrium $\mb{p}$ along the bifurcation curve $h_p$. Parameter values are $(m,d)=(0.0463361,0.8740509)$ in panels (a1)-(a2),  $(m,d)=(0.0463364,1.0390163)$ in panels (b1)-(b2), and $(m,d)=(0.0463366,1.2995479)$ in panels (c1)-(c2). The other parameters values are the same as in Fig.~\ref{fig:pulsodeonda}.}
	\label{fig:hom_p}
\end{figure}

Fig.~\ref{fig:hom_p} shows three different homoclinic orbits to the equilibrium $\mb{p}$ in the left-hand side column, and their corresponding time series in the right-hand side column. The values of parameters $(m,d)$ for each case correspond to those marked as $\diamondsuit,$ $\Box$ and $\times$ along the curve $h_p$ in Fig.~\ref{fig:bifurcacion}. 
In the left-hand side column of Fig.~\ref{fig:hom_p}, the homoclinic trajectories develop a rotational movement near $\mb{p}$ before making a large excursion and returning to $\mb{p}$; see the sequence of panels (a1)-(b1)-(c1). The amplitude of these oscillations increases as the point $(m,d)$ moves to the right along the curve $h_p$. As a result, the corresponding wave pulses in panels (a2)-(b2)-(c2) feature an initial transient with increasingly larger oscillations ---as $(m,d)$ moves to the right along $h_p$--- around the equilibrium values; in each case, this culminates in a large pulse before decaying back to the rest state.
Indeed, the initial pattern of smaller amplitude oscillations of each wave takes most of a long interval of values of $z\in]0,1[$ (i.e, it is a ``slow" build-up in terms of $z$); while the large amplitude pulse occurs in a smaller interval (of order $10^{-4}$) of parameter $z$ (i.e, a ``fast" discharge).
Further, notice that both  state variables $\mb u$ and $\mb v$ tend to increase and decrease simultaneously along any given traveling pulse.

The existence of the homoclinic orbit to $\mb{p}$ implies the presence of chaotic dynamics in \eqref{sist}. Let us now state the main reasons for this claim. For any $(m,d)$ in a neighbourhood of the curve $h_p$, the linearization of \eqref{sist} at the equilibrium $\mb{p}$ has one (stable) eigenvalue $\lambda^s<0$ and three (unstable) eigenvalues $\lambda_{1,2}^u\in\mathbb{C}$, and $\lambda_3^u>0$. In particular, $\lambda_{1,2}^u$ are complex conjugate with positive real part ${\rm Re}(\lambda_{1,2}^u)>0$. The equilibrium $\mb{p}$ is called a {\em saddle-focus}. Fig.~\ref{fig:vpchangedorientation} shows all the possible values of the eigenvalues of $\mb{p}$ (in the complex plane) along the computed segment of the homoclinic bifurcation curve $h_p$. Namely, as parameters $(m,d)$ are allowed to vary along the computed segment of the curve $h_p$ in Fig.~\ref{fig:bifurcacion}, each eigenvalue of $\mb{p}$ traces out a curve segment whose plots are shown in Fig.~\ref{fig:vpchangedorientation}. Among the unstable eigenvalues, the pair $\lambda_{1,2}^u$ are the closest to the imaginary axis ${\rm Re}(\lambda)=0$; hence, we say that $\lambda_{1,2}^u$ are the {\em leading unstable} eigenvalues. In this setting, if we define the so-called {\em saddle quantity} as $\sigma_1=\lambda^s+{\rm Re}(\lambda_{1,2}^u)$, Shilnikov's theorems~\cite{guckenheimer,kuznetsov,shil65,shil70} state that if $\sigma_1>0$, the homoclinic bifurcation is {\em simple} or {\em mild}. In this simple Shilnikov homoclinic bifurcation, a single (repelling) periodic orbit bifurcates from the homoclinic orbit on one side of the curve $h_p$. On the other hand, if $\sigma_1<0$, the homoclinic bifurcation is chaotic and gives rise to a wide range of complicated behavior in phase space. More specifically, one can find horseshoe dynamics in return maps defined in a neighbourhood of the homoclinic orbit. The suspension of the Smale horsehoes forms a hyperbolic invariant chaotic set which contains countably many periodic orbits of saddle-type. The horseshoe dynamics is robust under small parameter perturbations, i.e., the chaotic dynamics persist even when the homoclinic connection is broken; see~\cite{guckenheimer,kuznetsov}. The segments labeled as $h_p^s$ and $h_p^c$ in Fig.~\ref{fig:bifurcacion} correspond to simple and chaotic regimes, respectively, and are separated by the Belyakov point $B$ where $\sigma_1=0$~\cite{shilnikov,bel84}. (Actually, at $(m,d)=B$, we have ${\rm Re}(\lambda^u_{1,2})=|\lambda^s|\approx 0.6654466$). Likewise, in Fig.~\ref{fig:vpchangedorientation}, the segments $h_p^s$ and $h_p^c$ along the curves for $\lambda_{1,2}^u$ correspond to $\sigma_1>0$ (simple) and $\sigma_1<0$ (chaotic), respectively, and are separated by the point labeled as $B$ where $\sigma_1=0$. This same Belyakov transition is shown for the corresponding $\lambda^s$ value as well (and is also labeled as $B$).

	\begin{figure}
	\centering
	\includegraphics[scale=1]{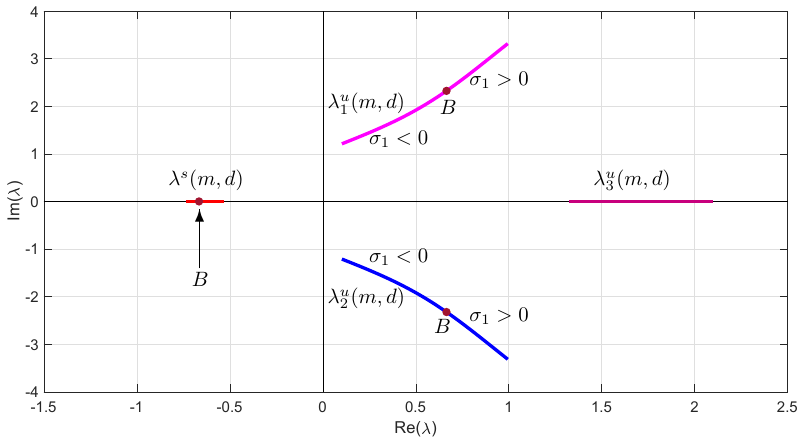}
	\caption{Eigenvalues of the Jacobian matrix $DX(\mb{p})$ continued for every $(m,d)$ along the homoclinic bifurcation curve $h_p$ in Fig.~\ref{fig:bifurcacion}.}
	\label{fig:vpchangedorientation}
\end{figure}

The bifurcation picture near the curve $h_p$ in Fig.~\ref{fig:bifurcacion} is just a partial representation of the full complexity one may encounter in this region of parameter space. Indeed, the saddle periodic orbits associated with the invariant chaotic set may also undergo further bifurcations such as period-doubling and torus bifurcations~\cite{guckenheimer,kuznetsov}. Moreover, the presence of the chaotic $h_p^c$ bifurcation and that of the Belyakov point $B$ imply a very complicated structure (not shown) of infinitely many saddle-node and period-doubling bifurcations of periodic orbits as well as of subsidiary $n$-homoclinic orbits. Fig.~\ref{fig:2homp} shows a 2-homoclinic orbit to $\mb{p}$ (in panel (a1)) and a 4-homoclinic orbit to $\mb{p}$ (in panel (b1)), as well as their corresponding time series in panels (a2) and (b2), respectively.
In general, $n$-homoclinic orbits are characterized by making $n-1$ close passes near the equilibrium before closing up to form the connection; see panels (a1) and (b1). As a consequence, the corresponding traveling wave develops $n$ pulses before setting down to the steady state values; see the 2-pulse and 4-pulse waves in panels (a1) and (b1), respectively.
Moreover, for each of these subsidiary $n$-homoclinic orbits, the system exhibits horseshoe dynamics and chaos as in the original homoclinic scenario.

	\begin{figure}
	\centering
	\includegraphics[width=\textwidth]{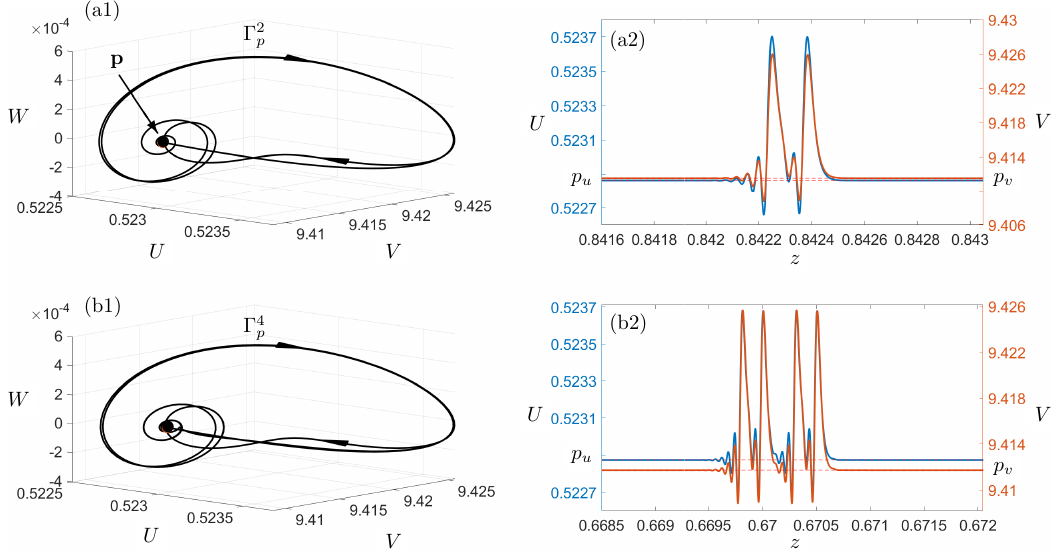}
	\caption{Panel (a1) shows $\Gamma_p^2$, the 2-homoclinic orbit to $\mb{p}$, while panel (a2) shows its time series of $U$ and $V$ associated with $\Gamma_p^2$. Similarly, panels (b1)-(b2) show a 2-homoclinic orbit and its associated 4-pulse wave, respectively. Parameter values are the same as in Fig.~\ref{fig:pulsodeonda} except for $(m,d)=(0.0463361, 1.1533894)$ in panels (a1)-(a2) and $(m,d)=(0.0463365, 1.1683875)$ in panels (b1)-(b2).}
	\label{fig:2homp}
\end{figure}

	\begin{figure}
	\centering
	\includegraphics[scale=1]{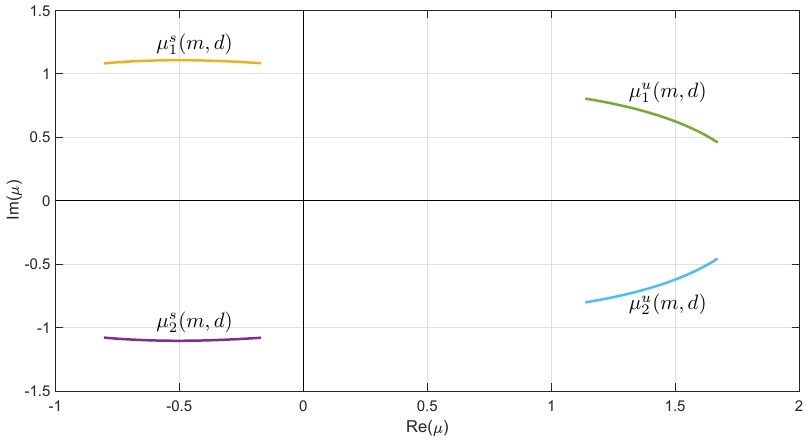}
	\caption{Eigenvalues $\mu_1^s(m,d)$, $\mu_2^s(m,d)$, $\mu_3^u(m,d)$ and $\mu_4^u(m,d)$ of $DX(\mb{q})$ continued for every $(m,d)$ along the homoclinic bifurcation curve $h_q$ in Fig.~\ref{fig:bifurcacion}.}
	\label{fig:vphomps}
\end{figure}

As for the homoclinic bifurcation $h_q$ at the equilibrium $\mb{q}$, the associated Jacobian matrix of \eqref{sist} at $\mb{q}$ has a pair of complex-conjugate stable eigenvalues $\mu_{1,2}^s\in\mathbb{C}$ and a pair of complex-conjugate unstable eigenvalues $\mu_{1,2}^u\in\mathbb{C}$, with ${\rm Re}(\mu_{1,2}^s)<0$ and ${\rm Re}(\mu_{1,2}^u)>0$. The structure of the eigenvalues of $\mb{q}$ as a function of $(m,d)\in h_q$ is shown in Fig.~\ref{fig:vphomps}. We say that $\mb{q}$ is a {\em focus-focus} or {\em bi-focus}. The resulting homoclinic orbit $\Gamma_q$ features a spiral-type convergence to $\mb{q}$ as $z\rightarrow\pm\infty$.
Fig.~\ref{fig:hom_q} shows three different examples of such homoclinic orbit in the left column, and their respective time series in the right one. The values of parameters $(m,d)$ for each case correspond to those marked as $\diamondsuit,$ $\Box$ and $\times$ along the curve $h_q$ in Fig.~\ref{fig:bifurcacion}. 
 In the {left-hand side} column of Fig.~\ref{fig:hom_q}, the amplitude of the oscillations increases as the point $(m,d)$ moves to the right along the curve $h_q$. As a result, the corresponding wave pulses in panels (a2)-(b2)-(c2) develop more oscillations ---as $(m,d)$ moves to the right along $h_q$--- before converging to the equilibrium values as $z\rightarrow\infty$. The spirals and oscillations that are visible in panels (a1)-(b1)-(c1) and in panels (a2)-(b2)-(c2), respectively, are associated with the stable eigenvalues $\mu_{1,2}^s$ of $\mb{q}$. There is another set of oscillations as $z\rightarrow-\infty$ which are associated with the unstable eigenvalues $\mu_{1,2}^u$; however, since ${\rm Im}(\mu_{1,2}^u)<{\rm Im}(\mu_{1,2}^s)$ (see Fig.~\ref{fig:vphomps} again), these spirals are relatively less pronounced and hard to see in Fig.~\ref{fig:hom_q}.
 Nevertheless, like the case of the homoclinic orbit to $\mb{p}$, here both $\mb u$ and $\mb v$ tend to increase and decrease simultaneously along any given traveling pulse.

	\begin{figure}
	\includegraphics[trim = 20mm 0mm 11mm 0mm, clip, width=\textwidth]{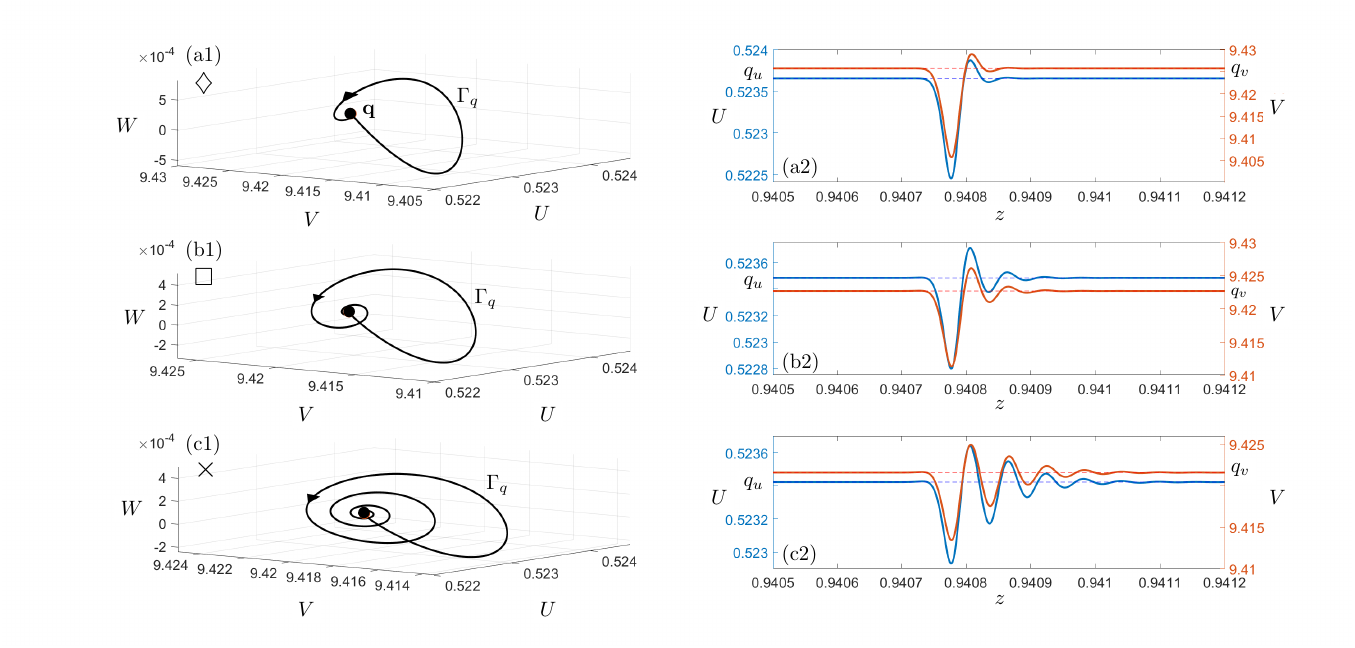}
	\caption{Homoclinic orbits to the equilibrium $\mb{q}$ along the bifurcation curve $h_q$. Parameter values are $(m,d)=(0.0463362,1.2091095)$ in panels (a1)-(a2),  $(m,d)=(0.0463365,1.1092036)$ in panels (b1)-(b2), and $(m,d)=(0.0463366,1.0724625)$ in panels (c1)-(c2). The other parameters values are the same as in Fig.~\ref{fig:pulsodeonda}.}
	\label{fig:hom_q}
\end{figure}

The homoclinic bifurcation at the focus-focus equilibrium $\mb{q}$ induces chaotic dynamics for every $(m,d)\in h_q$. Indeed, the presence of the homoclinic orbit $\Gamma_q$ to a focus-focus equilibrium is accompanied by horseshoe dynamics in cross sections near $\Gamma_q$ and, hence, an infinite number of saddle periodic orbits in a neighbourhood of $\Gamma_q$~\cite{kuznetsov,shilnikovruso}. Furthermore, 
in this setting, the {\em saddle quantity} is defined as $\sigma_2={\rm Re}(\mu_{1,2}^s)+{\rm Re}(\mu_{1,2}^u)$. Since $\sigma_2>0$ for every $(m,d)\in h_q$, it follows that there are no stable periodic orbits near $\Gamma_q$~\cite{foco-foco,sandstede}. 

In sum, any solution in a neighborhood of either $\Gamma_p$ (in the chaotic case) or $\Gamma_q$ tends to behave erratically and presents sensitive dependence to initial conditions. The corresponding orbit in the four-dimensional phase space of \eqref{sist} spends a long transient visiting a strange hyperbolic invariant set before converging to an attractor. Hence, any bounded solution of \eqref{sist} passing near either $\Gamma_p$ (in the $h_p^c$ side of the bifurcation) or $\Gamma_q$ is associated with a chaotic traveling wave~\cite{shilnikovruso}.

Fig.~\ref{fig:doblehom}(a) shows both homoclinic orbits $\Gamma_p$ and $\Gamma_q$ coexisting in phase space, while Fig.~\ref{fig:doblehom}(b1) and Fig.~\ref{fig:doblehom}(b2)  show the time series of $U$ and $V$ associated with either trajectory. This special configuration occurs when the bifurcation curves $h_p^c$ and $h_q$ cross each other at $(m,d)\approx(0.046336476, 1.11668)$; see the bifurcation diagram of Fig.~\ref{fig:bifurcacion}.
While this intersection point is not a new bifurcation, at these parameter values both classes of homoclinic orbits, $\Gamma_q$ and $\Gamma_p$, coexist in phase space. Moreover, one obtains the coexistence of both chaotic invariant sets (each associated with one of the homoclinic trajectories) and, hence,  the corresponding erratic behavior and sensitive dependence on initial conditions of nearby solutions.


	\begin{figure}
	\includegraphics[width=\textwidth]{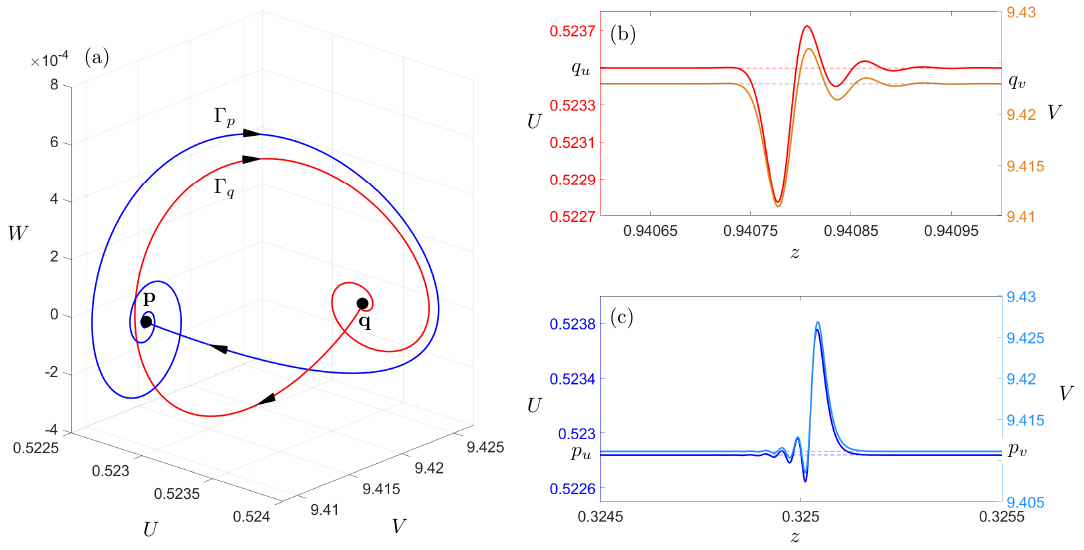}
	\caption{Panel (a) shows a projection of $\Gamma_s$ and $\Gamma_p$ onto the $UVW$ space, when $(m,d)\approx(0.046336476, 1.11668)\in  h_p^c\cap h_q$. Meanwhile, panel (b1) (resp. (b2)) shows the time series of $U$ and $V$ rendered in different color tones, associated with $\Gamma_s$ (resp. $\Gamma_p$). The other parameter values are the same as in Fig.~\ref{fig:pulsodeonda}.}
	\label{fig:doblehom}
\end{figure}

\section{Periodic orbits and wave trains}
\label{sec:cycles}

In this section we study the limit cycles existing in \eqref{sist}, their bifurcations, and their consequences for the nature of wave trains.

\subsection{Period doubling phenomena}
\label{sec:cycles-pd}

	\begin{figure}
	\includegraphics[trim = 20mm 0mm 18mm 0mm, clip,width=\textwidth]{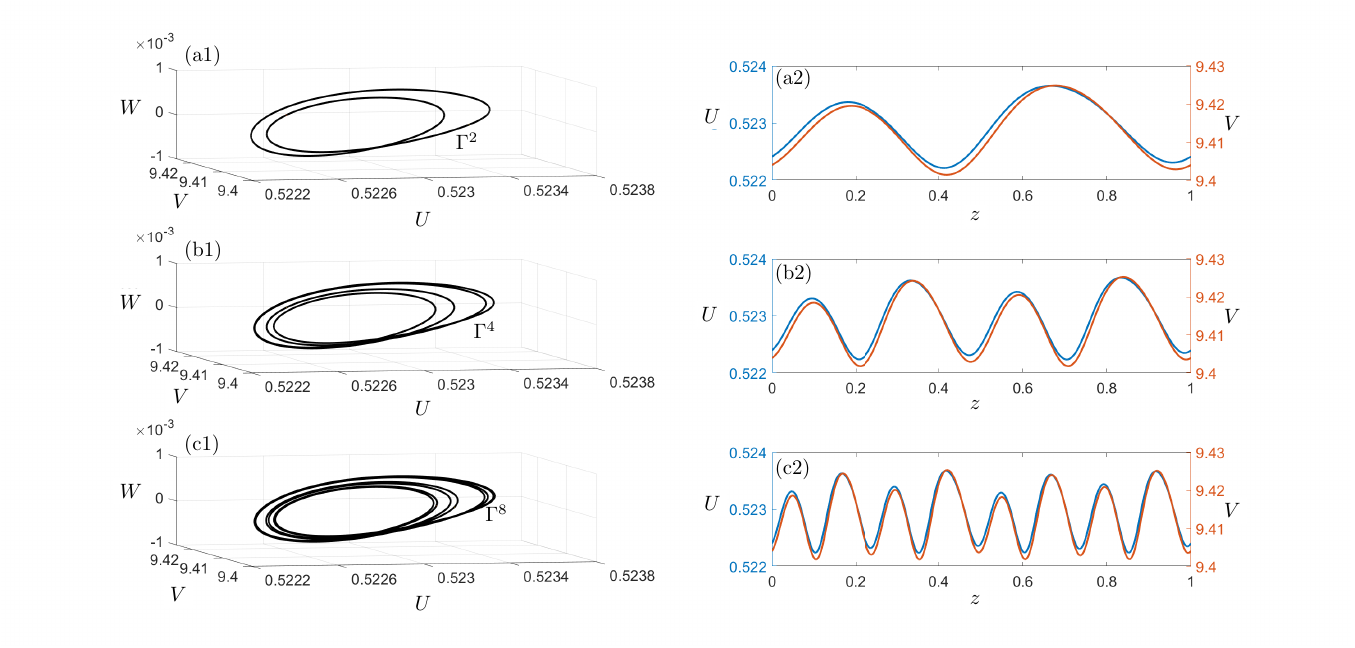}
	\caption{The different periodic orbits (panels (a1), (b1) and (c1)) and associated wave trains (panels (a2), (b2) and (c2))
	emerging from successive period doubling bifurcations. While the periods in the time series are uniformly rescaled to 1 for computational purposes, the actual periods of the cycles are $T=12.1829$ for $\Gamma^2$, $T=24.3659$ for $\Gamma^4$, and $T=48.7318$ for $\Gamma^8$. Parameter values are $d=1.469369$ (in panels (a)), $d=1.4607309$ (in panels (b)), and $d=1.4590971$ (in panels (c)), with $m=0.0463362$ fixed. The other parameters as are as in Fig.~\ref{fig:pulsodeonda}.}
	\label{fig:allperiods}
\end{figure}

Let us consider the periodic orbit $\Gamma$ which originates at the supercritical Hopf bifurcation $H_p^-$ and track its successive bifurcations as parameter $d$ is decreased and $m=0.0463358$ remains fixed. When $(m,d)$ crosses the $PD$ curve from region II to III, the cycle $\Gamma$ undergoes a period doubling bifurcation. As $(m,d)$ enters region III, $\Gamma$ changes its stability and a secondary limit cycle $\Gamma^2$ appears with approximately twice the period of $\Gamma$. As parameter $d$ is further decreased, additional period doubling events occur (not shown in Fig.~\ref{fig:bifurcacion}). This is illustrated in Fig.~\ref{fig:allperiods}. Periodic orbits $\Gamma^2$, $\Gamma^4$ and $\Gamma^8$ of periods 2, 4, and 8 times that of $\Gamma$, respectively, are shown in panels (a1)-(b1)-(c1). Panels (a2)-(b2)-(c2) show one period of the corresponding time series of $U$ and $V$. Here, the actual periods of the solutions are rescaled to $T=1$ for visualization and computational purposes~\cite{doe-book}. As a consequence, as parameter $d$ decreases and system \eqref{sist} undergoes this sequence of period doubling bifurcations, the associated wave trains in panels (a2)-(b2)-(c2) display periodic  patterns with doubling periods.

\subsection{Transition from wave trains to wave pulses}
\label{sec:cycles-shilnikov}

	\begin{figure}
	\centering
	\includegraphics[scale=0.8]{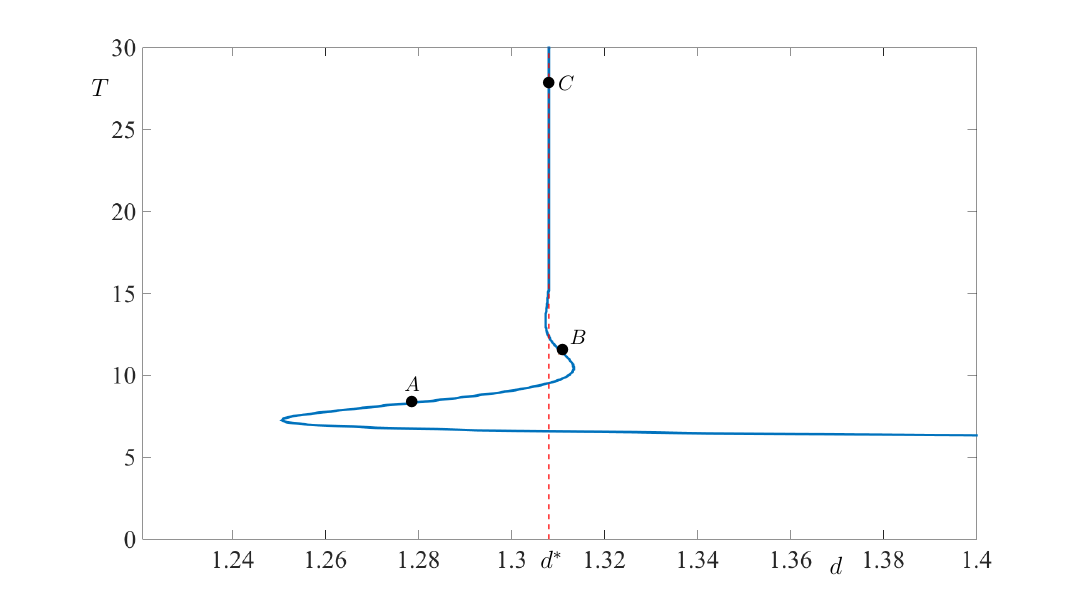}
	\caption{Bifurcation curve of the period $T$ of periodic orbits with respect to $d$, near the homoclinic bifurcation at $\mb{p}$. Here $m=0.0463358$ is fixed and the other parameters as are as in Fig.~\ref{fig:pulsodeonda}.}
	\label{fig:periodoend}
\end{figure}

It is essential to highlight that when the point $(m,d)$ crosses the PD curve from region II to region III, the cycle $\Gamma$ does not disappear but just changes its stability. Fig.~\ref{fig:periodoend} shows the graph of the period $T$ of this cycle as a function of $d$. The bifurcation curve oscillates around the critical value $d^*\approx1.3080156$ for which the homoclinic bifurcation $h_q$ to $\mb{q}$ occurs. The amplitude of the oscillations decreases rapidly as the homoclinic limit is approached when $d$ tends to $d^*$; see~\cite{kuznetsov} and references therein. Indeed, the ``snaking" behavior of the bifurcation curve is typical of the main branch of periodic orbits near chaotic saddle-focus and focus-focus homoclinic bifurcations~\cite{foco-foco,wiggins}. At each of the infinitely many folds of the curve, a pair of periodic orbits is created via a saddle-node bifurcation of limit cycles. Some of the periodic orbits in this branch may further undergo period-doubling bifurcations changing their stability along the bifurcation curve. Fig.~\ref{fig:cambiociclo} shows three such periodic orbits, labelled as $\Gamma_A$, $\Gamma_B$, and $\Gamma_C$, respectively, corresponding to the points $A$, $B$ and $C$ in Fig.~\ref{fig:periodoend}. As $d$ approaches $d^*$, the cycles pass increasingly closer to $\mb{q}$ (see the sequence of panels (a1)-(b1)-(c1) in Fig.~\ref{fig:cambiociclo}). As a result, one obtains wave trains which spend longer transients close to the equilibrium values (see the sequence of panels (a2)-(b2)-(c2) in which the period $T$ of each cycle is rescaled to 1). Hence, one can think of the homoclinic orbit $\Gamma_q$ (and its corresponding wave pulse) as the limit of this sequence of periodic orbits (resp. wave trains) of increasing period as $d\rightarrow d^*$. Furthermore, each of the periodic orbits bifurcated from the period doubling phenomena in {subsection}~\ref{sec:cycles-pd} may also increase their periods and undergo a convergence to $n$-homoclinic orbits in a similar fashion. Some of these secondary homoclinic bifurcations are mentioned before in \S\ref{sec:homoclinic} and shown in Fig.~\ref{fig:2homp}.

	\begin{figure}
	\includegraphics[width=\textwidth]{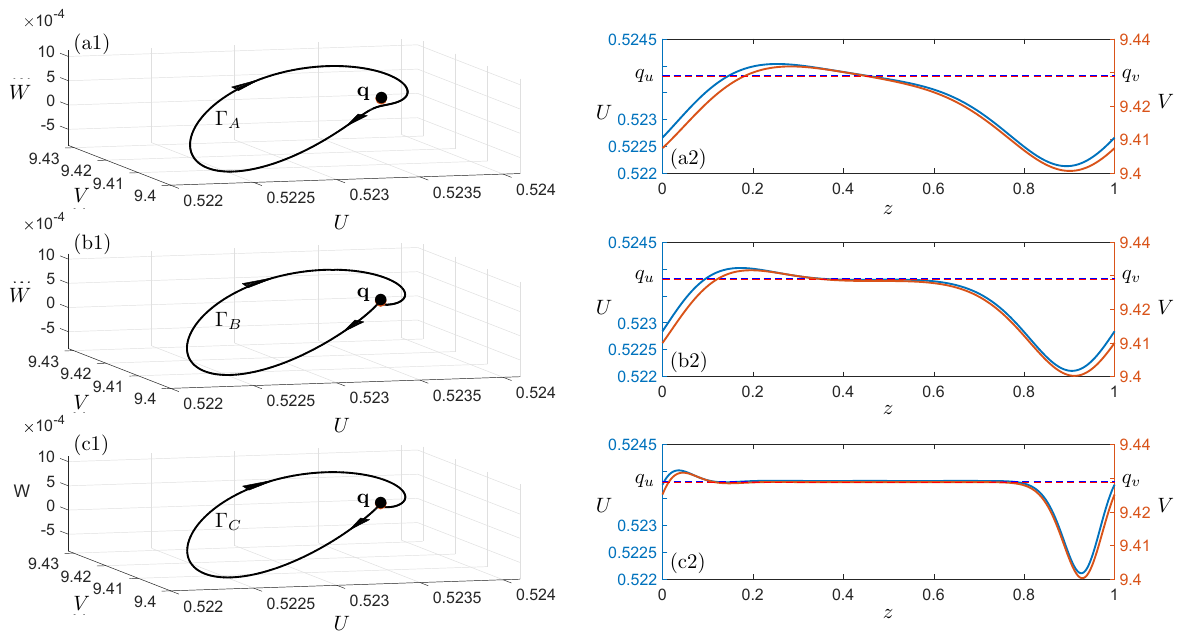}
	\caption{The different periodic orbits $\Gamma_A$, $\Gamma_B$ and $\Gamma_C$ (in panels (a1), (b1) and (c1)) and associated wave trains (panels (a2), (b2) and (c2)). While the periods in the time series are uniformly rescaled to 1 for computational purposes, the actual periods of the cycles are $T=8.2967$ for $\Gamma_A$, $T=11.4382$ for $\Gamma_B$, and $T=27.8801$ for $\Gamma_C$. Parameter values are $d=1.2785$ (in panels (a)), $d=1.3107$ (in panels (b)), and $d=1.3080$ (in panels (c)), with $m=0.0463358$ fixed. The other parameters as are as in Fig.~\ref{fig:pulsodeonda}.}
	\label{fig:cambiociclo}
\end{figure}

\section{Heteroclinic connections and wave fronts}
\label{sec:heteroclinic}

Fig.~\ref{fig:heteroI}(a) shows a heteroclinic orbit, labeled as $\Gamma_{q,p}$,  obtained for $(m,d)=(0.0463358, 2.4)$ in region I. The heteroclinic connection is oriented from $\mb{q}$ to $\mb{p}$ as parameter $z$ ---that which parametrizes the curve--- is increased. The resulting time series are shown in Fig.~\ref{fig:heteroI}(b), and they correspond to a wave front traveling from the steady state $\mb{q}$ that decays exponentially to $\mb{p}$ in synchronized oscillatory fashion. This non-monotonic behavior is explained by the presence of a pair of stable complex-conjugate eigenvalues (with negative real part) of $\mb{p}$ when parameters $(m,d)$ are in region I. 
The connecting orbit $\Gamma_{q,p}$ lies in the intersection of the global invariant manifolds $W^s(\mb{p})$ and $W^u(\mb{q})$. Namely, $\Gamma_{q,p}$ is contained in the two-dimensional unstable manifold $W^u(\mb{q})$ ---represented in Fig.~\ref{fig:heteroI}(a) as a transparent red surface--- and approaches $\mb{p}$ along its three-dimensional stable manifold $W^s(\mb{p})$ (not shown).
Moreover, since $W^s(\mb{p})$ and $W^u(\mb{q})$ are, respectively, three and two-dimensional immersed smooth manifolds 
in $\mathbb{R}^4$, one may anticipate that the intersection $W^s(\mb{p})\cap W^u(\mb{q})$ will generically be transverse~\cite{hirsch}. Indeed, this is in agreement with the fact that $W^s(\mb{p})\cap W^u(\mb{q})$ corresponds to a one-dimensional object in $\mathbb{R}^4$. Hence, one may reliably expect the heteroclinic orbit $\Gamma_{q,p}$ and, hence, its associated wave front to persist under small parameter variations: The resulting wave front may vary the amplitude of its oscillations and the actual asymptotic values, but the qualitative behavior of the traveling front remains unaltered throughout.

	\begin{figure}
	\includegraphics[width=\textwidth]{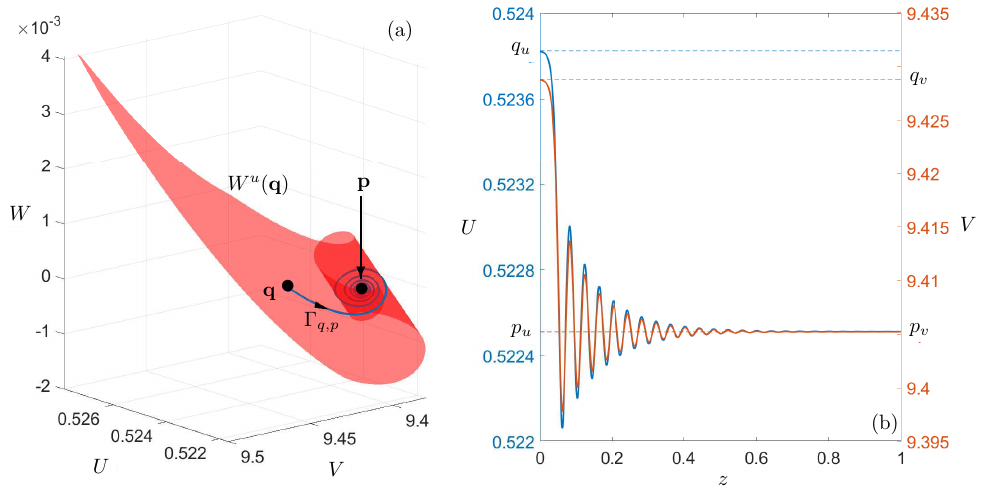}
	\caption{The connecting orbit $\Gamma_{q,p}$ lies in the intersection of the global invariant manifolds $W^s(\mb{p})$ and $W^u(\mb{q})$ in panel (a). In panel (b) the associated wave front travels from the steady state $\mb{q}$ and its amplitude decays exponentially fast to $\mb{p}$ showing oscillations. Parameter values are $(m,d)=(0.0463358,2.4)$ and the other parameters remain fixed as in Fig.~\ref{fig:pulsodeonda}.}
	\label{fig:heteroI}
\end{figure}

As parameter $(m,d)$ crosses the $H_p^-$ curve from region I to region II, a pair of stable eigenvalues of $\mb{p}$ cross the imaginary axis and become unstable; in the process, a limit cycle branches from $\mb{p}$ in a supercritical Hopf bifurcation. Hence, in region II, $W^s(\mb{p})$ is a one-dimensional manifold and the connection $\Gamma_{q,p}$ does not exist. Rather, it is replaced by a heteroclinic orbit that joins $\mb{q}$ to the bifurcated cycle.
Fig.~\ref{fig:heteroEtoP}(a) shows the bifurcated periodic orbit, labeled as $\gamma$ after $(m,d)$ has entered region II from region I.
The unstable manifold $W^u(\mb{q})$ (red surface) rolls up around $\gamma$ and intersects the three-dimensional stable manifold $W^s(\gamma)$ (not shown) transversally along a heteroclinic orbit, labeled as $\Gamma_{q,\gamma}$. This heteroclinic connection is associated with the traveling wave shown in Fig.~\ref{fig:heteroEtoP}(b); this is a front transitioning from the steady state at $\mb{q}$ into a periodic pattern around $\mb{p}$. The heteroclinic orbit $\Gamma_{q,\gamma}$ (and its traveling front) is preserved in an open subset of region II, and it disappears when $(m,d)$ crosses the PD curve towards region III as $\gamma$ loses its stability in a period doubling bifurcation.

	\begin{figure}
	\includegraphics[width=\textwidth]{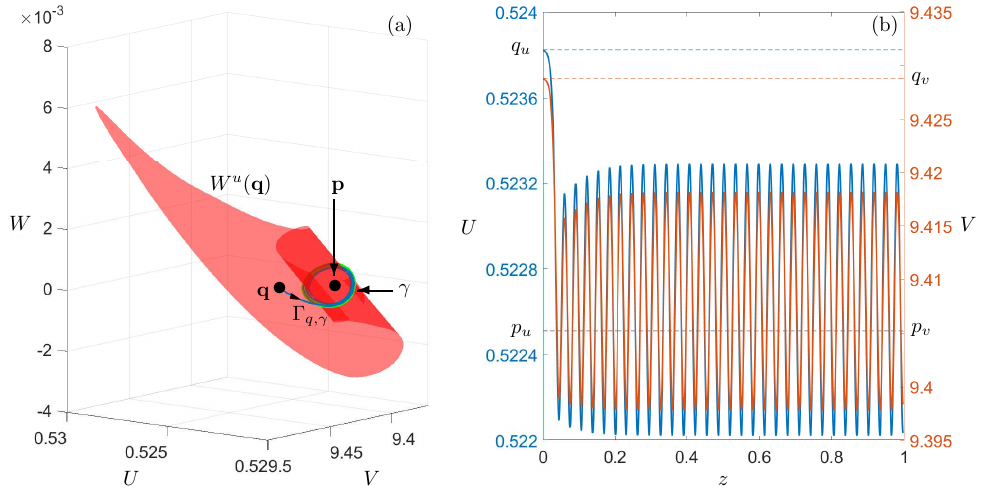}
	\caption{The connecting orbit $\Gamma_{q,\gamma}$ lies in the intersection of the global invariant manifolds $W^s(\gamma)$ and $W^u(\mb{q})$ in panel (a). The associated wave front travels from the steady state $\mb{q}$ and adopts a periodic behavior oscillating around the equilibrium values of $\mb{p}$ in panel (b). Parameter values are $(m,d)=(0.0463358,1.7)$ and the other parameters remain fixed as in Fig.~\ref{fig:pulsodeonda}.}
	\label{fig:heteroEtoP}
\end{figure}

If $(m,d)$ is in an open subset of regions III, IV, and V, one can find a wave front traveling from $\mb{p}$ to $\mb{q}$. This front travels in the opposite direction to that in Fig.~\ref{fig:heteroI} and, hence, it corresponds to a third kind of wave.
 The wave front traveling from $\mb{p}$ to $\mb{q}$ is shown in Fig.~\ref{fig:heteroII}(b). The wave begins at the steady state $\mb{p}$ with oscillations of increasing amplitude until it settles at $\mb{q}$.
This front corresponds to an intersection of the manifolds $W^s(\mb{q})$ and $W^u(\mb{p})$ forming a heteroclinic orbit in the phase space of \eqref{sist}; Fig.~\ref{fig:heteroII}(a) shows the heteroclinic orbit (labelled as $\Gamma_{p,q}$) and the two-dimensional manifold $W^s(\mb{q})$ of $\mb{q}$ as a transparent blue surface. The connection $\Gamma_{p,q}$ is an orbit in the three dimensional unstable manifold $W^u(\mb{p})$ which lies on $W^s(\mb{q})$ to converge to $\mb{q}$.

	\begin{figure}
	\includegraphics[width=\textwidth]{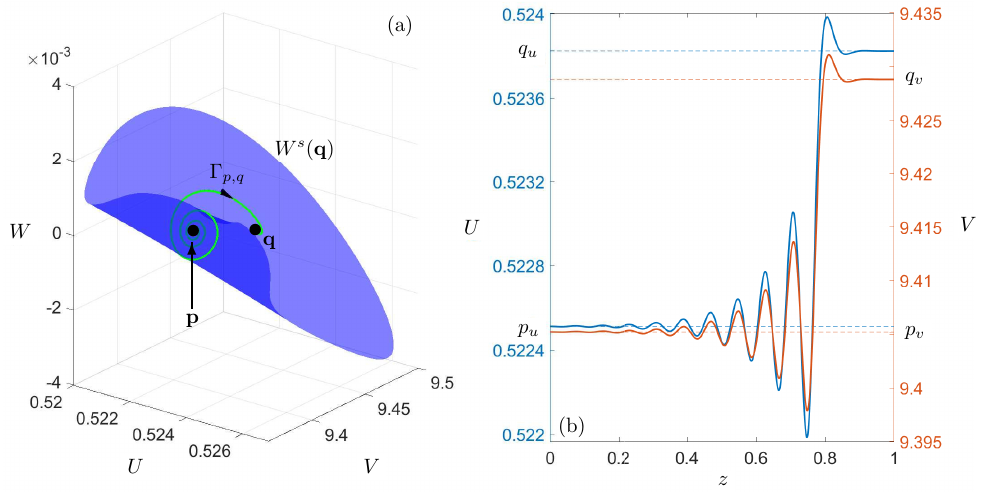}
	\caption{The connecting orbit $\Gamma_{p,q}$ lies in the intersection of the global invariant manifolds $W^u(\mb{p})$ and $W^s(\mb{q})$ in panel (a).  In panel (b) the associated wave front travels from the steady state $\mb{p}$ and its amplitude increases exponentially fast before settling down at $\mb{q}$. Parameter values are $(m,d)=(0.0463358,1.4)$ and the other parameters remain fixed as in Fig.~\ref{fig:pulsodeonda}.}
	\label{fig:heteroII}
\end{figure}

\section{Multiple wave fronts at the focus-focus homoclinic bifurcation}
\label{sec:focusfocus}

In \S\ref{sec:homoclinic}, we described the complicated dynamics that can be found near the focus-focus homoclinic bifurcation $\Gamma_q$ when $(m,d)\in h_q$. One of the consequences of this fact is the appearance of multiple wave fronts of type $\Gamma_{p,q}$ (described in \S\ref{sec:heteroclinic}) which coexist with the main wave pulse $\Gamma_q$. We explain this finding here by direct, close inspection of the invariant manifolds involved.

	\begin{figure}
	\includegraphics[trim = 15mm 0mm 0mm 0mm, clip,width=\textwidth]{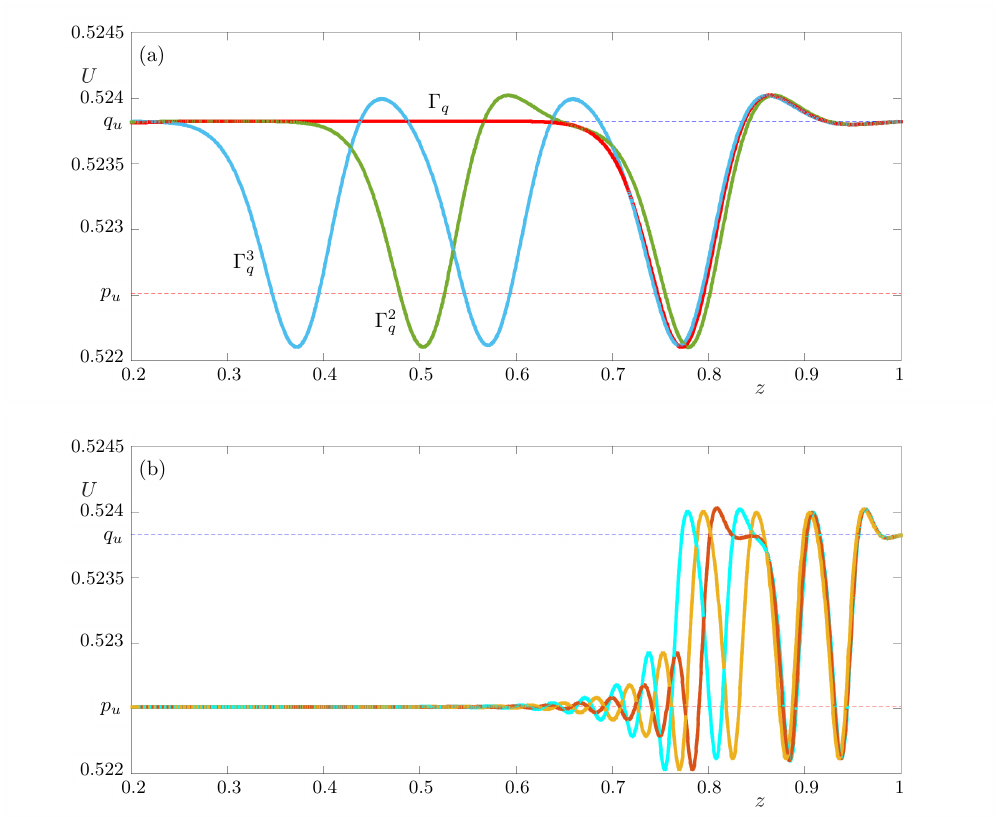}
	\caption{The stable manifold $W^s(\mb{q})$ of $\mb{q}$ projected onto the $UVW$ space in panel (a) when $(m,d)\approx(0.0463358,1.3080156)\in h_q$. The manifold $W^s(\mb{q})$ contains a family of coexisting heteroclinic orbits which join the equilibrium $\mb{p}$ to $\mb{q}$ as well as the primary focus-focus homoclinic connection $\Gamma_q$ to $\mb{q}$. One of such heteroclinic orbits is represented by the trajectory $\Gamma_{p,q}$. Panels (b) and (c) show enlargements near the equilibria $\mb{p}$ and $\mb{q}$, respectively.
	The other parameter values are the same as in Fig.~\ref{fig:pulsodeonda}.}
	\label{fig:variedadeshom}
\end{figure}

	\begin{figure}
	\centering
	\includegraphics[trim = 8mm 0mm 8mm 0mm, clip,scale=0.8]{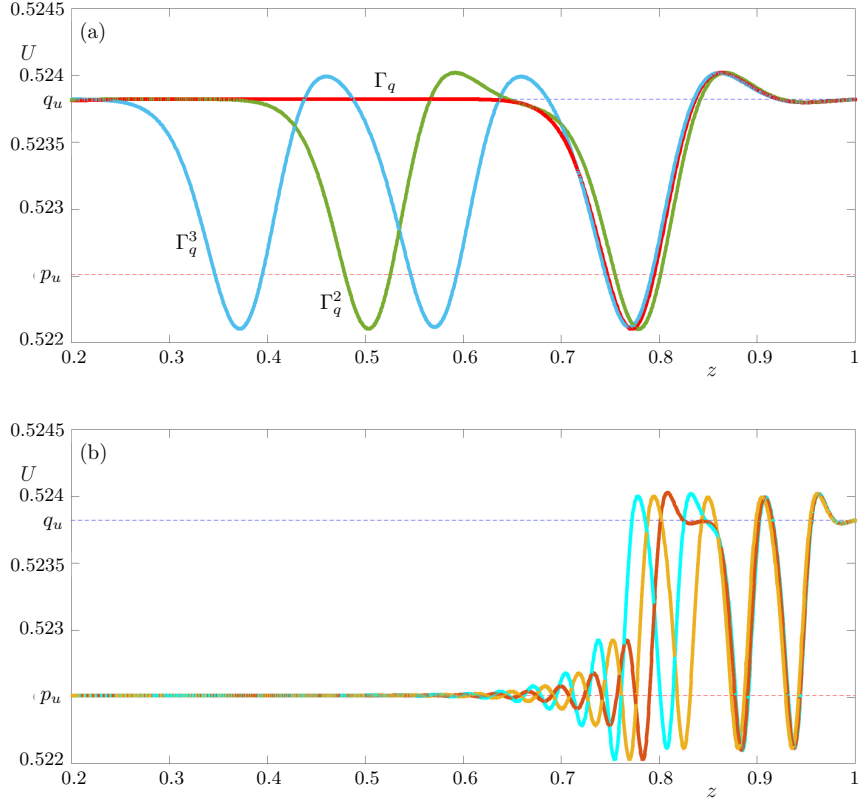}
	\caption{Profiles of $U$ associated with the wave pulses $\Gamma_q$, $\Gamma_q^2$ and $\Gamma_q^3$ in panel (a);  and three representative traveling fronts $\Gamma_{p,q}^1$, $\Gamma_{p,q}^2$, and $\Gamma_{p,q}^3$ in panel (b), when $(m,d)\approx(0.0463358,1.3080156)\in h_q$.  The other parameter values are the same as in Fig.~\ref{fig:pulsodeonda}.}
	\label{fig:heteroshom}
\end{figure}

Fig.~\ref{fig:variedadeshom} shows the projection of $W^s(\mb{q})$ onto the $UVW$ space, when 
$$(m,d)\approx (0.0463358,1.3080156)\in h_q.$$
 Also shown is the homoclinic orbit $\Gamma_{q}$ (red curve). The two-dimensional manifold $W^s(\mb{q})$ is rendered as a transparent blue surface. Some orbits in $W^s(\mb{q})$ lie in the three-dimensional unstable manifold $W^u(\mb{p})$ forming heteroclinic connections. In our computations, we detected {seven} heteroclinic orbits  contained in $W^u(\mb{p})\cap W^s(\mb{q})$.
 Most of these heteroclinic connections are very close to one another and very hard to distinguish from each other; we show one of them (cyan curve labeled as $\Gamma_{p,q}$) in Fig.~\ref{fig:variedadeshom}. 
Each of these {seven} heteroclinic orbits corresponds to a different wave front traveling from $\mb{p}$ to $\mb{q}$, which are present in the system at the same time.

Further, there are also 2- and 3-homoclinic orbits to $\mb{q}$ in $W^u(\mb{q})\cap W^s(\mb{q})$, coexisting with the primary homoclinic orbit and the heteroclinic connections in phase space; we opted to not show these subsidiary homoclinic orbits in Fig.~\ref{fig:variedadeshom} for visualization purposes.
Fig.~\ref{fig:heteroshom}(a) shows the profiles of $U$ associated with the primary wave pulse $\Gamma_q$ and the secondary 2- and 3-pulse waves (labeled as $\Gamma_q^2$ and $\Gamma_q^3$, respectively) associated with the secondary homoclinic orbits. In turn, Fig.~\ref{fig:heteroshom}(b) shows three representative traveling fronts associated with the family of heteroclinic orbits in phase space.
The existence of these families of wave fronts  when $(m,d)\in  h_q$ indicates that there must be a sequence of associated global bifurcations  as $(m,d)$ approaches the $h_q$ curve. At each of these bifurcation events, the manifolds $W^u(\mb{p})$ and $W^s(\mb{q})$ intersect tangentially in $\mathbb{R}^4$ along a (newly created) heteroclinic orbit. As $(m,d)$ moves closer to the $h_q$ curve, the intersection becomes transversal and the heteroclinic orbit persists under small parameter variations. Similar events happen in the case of the secondary homoclinic orbits in the intersection of $W^u(\mb{q})$ and $W^s(\mb{q})$, that explain the emergence of 2- and 3-pulse waves as parameters $(m,d)$ approach the $h_q$ curve.

The method to detect these connections is explained as follows. 
If $W^s(\mb{q})$ contains a heteroclinic orbit flowing from $\mb{p}$ to $\mb{q}$, such trajectory is approximated by a bounded solution contained in a family $\widehat{W}^s_{\delta,T}(\mb{q})$ of orbit segments passing sufficiently close to $\mb{p}$ and satisfying a two-point boundary value problem. Every solution in  $\widehat{W}^s_{\delta,T}(\mb{q})$ is continued up to an integration time $T$ (which is a free parameter in this continuation) and is parametrized by a unique location $\delta\in[0,1)$ in a fundamental domain; see~\cite{aguirrehom,shilnikov,numericalmanifold} for more details.
Indeed, if a heteroclinic  orbit exists, the two-parameter continuation procedure effectively stops as the integration time diverges. In practice, an approximation of such connecting orbit is obtained at some specific  $\delta=\delta^*\in[0,1)$ with a large integration time $T=-T^*$. A similar criterium can be used to detect secondary $n$-homoclinic orbits to $\mb{q}$ as orbit segments ending near $\mb{q}$ as the integration time $T$ diverges.
For instance, in Fig.~\ref{fig:variedadeshom} and Fig.~\ref{fig:heteroshom}, the fundamental domain $\delta\in[0,1)$ is divided into 13 sub-segments by the values $0<\delta_1<\delta_2<\ldots <\delta_{12}<1$. The heteroclinic connections correspond to $\delta_1\approx0.317181$, $\delta_2\approx0.317265$, $\delta_6\approx0.319579$, $\delta_8\approx0.319621$, $\delta_9\approx0.319904$, $\delta_{11}\approx0.319917$ and $\delta_{12}\approx0.331911$. On the other hand, we have the primary homoclinic orbit at $\delta_7\approx0.319616$, and four secondary homoclinic orbits at $\delta_{3}\approx0.317284$, $\delta_4\approx 0.317285$, $\delta_5\approx 0.317285$
and $\delta_{10}\approx0.319906.$

\section{The influence of propagation speed and diffusion ratio}
\label{sec:cd-plane}

	\begin{figure}
	\centering
	\includegraphics[trim = 0mm 0mm 5mm 0mm, clip,scale=0.8]{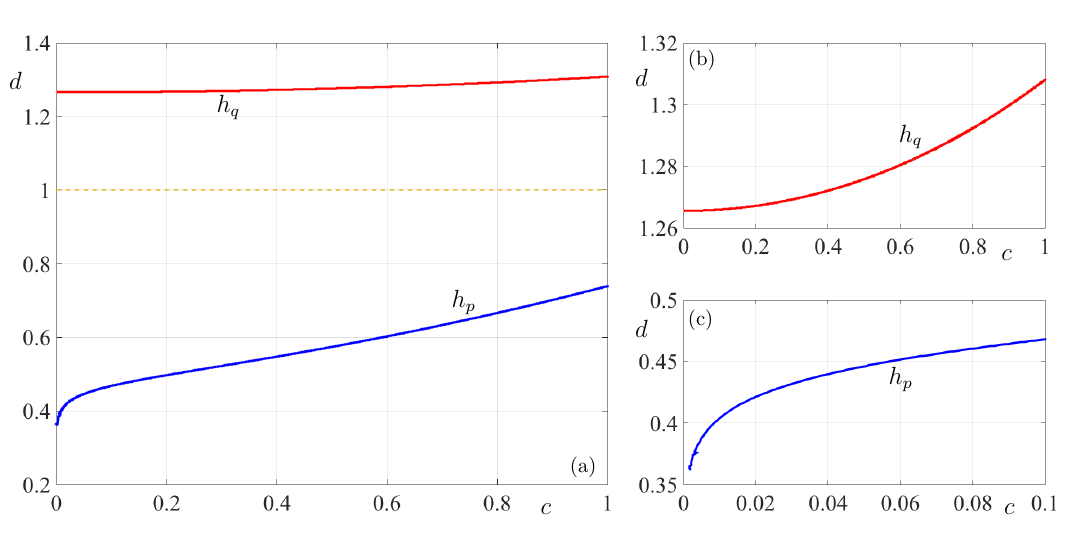}
	\caption{The homoclinic bifurcation curves $h_p$ and $h_q$ in the $(c,d)$-plane in panel (a). Panels (b) and (c) show enlargements near the curves $h_q$ and $h_p$, respectively (the apparently different shape of each curve is due to the different scales in each plot). Parameter $m=0.0463358$ and the other parameters are as in Fig.~\ref{fig:pulsodeonda}.}
	\label{fig:hompcd}
\end{figure}

The study and results reported so far in sections~\ref{sec:bif}--\ref{sec:focusfocus} were produced with fixed wave speed $c=1$. Here we ask ourselves if there is a minimum wave speed needed for the existence of some of the traveling waves we have found. To this end, we consider the homoclinic orbits $\Gamma_p$ and $\Gamma_q$ existing at the curves $h_p$ and $h_q$, respectively, when $m\approx0.0463358$, and continue them in parameters $c$ and $d$.
Fig.~\ref{fig:hompcd} shows the homoclinic bifurcation curves $h_p$ and $h_q$ in the $(c,d)$-plane. The existence of both homoclinic orbits is determined by a positive correlation between the wave speed $c$ and the diffusion ratio $d$; namely, as $c$ decreases, the wave pulses exist provided $d$ becomes sufficiently small. 

In the case of $h_p$, the relation between $c$ and $d$ in Fig.~\ref{fig:hompcd} is almost linear for $c\geq 0.1$. Indeed the curve $h_p$ can be approximated {as $d \approx 0.2925828c + 0.4333269,$ for} $0.1\leq c<1$, with a root mean square error $e=0.00521$.
In particular, the computed segment of the curve $h_p$ is located in the halfspace $d<1$. Hence, this kind of pulse wave with small wave speed $c<1$ occurs only if  $V$ propagates in a more efficient way than $U$.
On the other hand, as $c$ decreases below $0.01$, the diffusion ratio $d$ drops abruptly in a non linear way in the {form $d=\mathcal{O}(c^{1/2})$}; see Fig.~\ref{fig:hompcd}(c). As $c$ is further decreased, the continuation procedure loses precision and the last point where we get convergence of the numerical scheme is at $c_{\rm min}=0.0016767$. 
Fig.~\ref{fig:hompcchico} shows the homoclinic orbit $\Gamma_p$ when $(c,d)=(0.0479321, 0.4449391)$ in panel (a), and its corresponding time series in panel (b). In panel (a), the orbit $\Gamma_p$ performs many low-amplitude turns in $W^u_\text{loc}(\mb{p})$ before developing the long excursion. The corresponding wave in panel (b) shows a slow pattern (in terms of $z$) of small amplitude oscillations followed by a fast large amplitude pulse in a small interval (of order $10^{-4}$) of parameter $z$, similar to typical dynamic behaviors with different time scales.
Indeed, if the {relation $d=\mathcal{O}(c^{1/2})$} still holds for $c\to 0$, then both \eqref{edp2} and \eqref{sist} become singular as $c\rightarrow0$ and $d\rightarrow0$; while these systems in the singular limit may be studied with tools from geometric singular perturbation theory~\cite{fenichel}, this is beyond the scope of this work.  
 Nevertheless, solutions $u(x,t)=U(x+ct)$, $v(x,t)=V(x+ct)$ of \eqref{edp} as $c\to 0$ {and $d=D_1/D_2\rightarrow0$} correspond to {stationary waves} in which, effectively, only the  propagation of $V$ is observable in the length scale $x$.

	\begin{figure}
	\includegraphics[width=\textwidth]{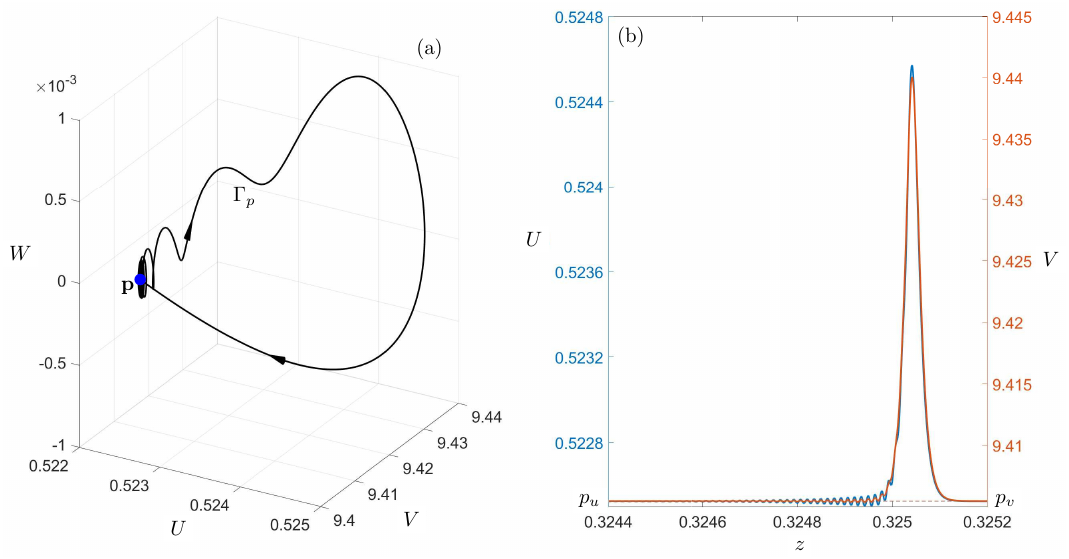}
	\caption{The homoclinic orbit $\Gamma_p$ to $\mb{p}$, when $(c,d)=(0.0479321, 0.4449391)$ in panel (a) and its wave profile in panel (b). The other parameter values are as in Fig.~\ref{fig:hompcd}.}
	\label{fig:hompcchico}
\end{figure}

As for the bifurcation curve $h_q$ in Fig.~\ref{fig:hompcd}(a)-(b), the dependence between $c$ and $d$ is approximately quadratic. That is, a homoclinic orbit to the focus-focus  $\mb{q}$ exists whenever $c$ and $d$ {satisfy $d \approx 0.0437881c^2 - 0.0016566c + 1.2656854$, for} $0\leq c\leq 1$; the root mean square error of this approximation is $e=6.8\times 10^{-5}$. 
In particular, the computed segment of the curve $h_q$ is located in the {half-space} $d>1$. Hence, this kind of pulse wave with small wave speed $c<1$ occurs only if $U$ propagates in a more efficient way than $V$.
Moreover, the numerical evidence suggests that a pulse wave exists for every $c>0$ arbitrarily small, i.e., there is no positive minimum value for the wave speed $c$. Indeed, the bifurcation curve $h_q$ can be continued down to $c_{\rm min}=0$ with $d_{\rm min}\approx 1.2656854$. (In particular, the value $d_{\rm min}>0$ prevents \eqref{edp2} and \eqref{sist} to become singular, unlike the case of $h_p$).
 In the limit as $c\rightarrow0$, the resulting wave pulse corresponds to a stationary solution of \eqref{edp} where both state variables $U$ and $V$} propagate with diffusion {ratio $D_1/D_2= d_{\rm min}$.

\section{Traveling waves restricted to invariant planes}
\label{sec:planes}

System \eqref{sist} has two invariant planes given by
$$
	\Pi_U=\{(U,V,W,R)\in \mathbb{R}^4 \ : \ V=R=0 \},
$$
and
$$
	\Pi_V=\{(U,V,W,R)\in \mathbb{R}^4 \ : \ U=W=0 \}.
$$
Note that the origin $\mb{p}_0\in \Pi_U\cap \Pi_V$. 
The restriction of \eqref{sist} to $\Pi_U$ is given by
\begin{eqnarray}
	X_U:\left\{\begin{array}{l}
	\dfrac{dU}{dz}=W,
	\\[1.5ex]
	\dfrac{dW}{dz}=\dfrac{1}{d}\left(cW-sU^2(U-m)(1-U)\right).
	\end{array}
	\right.
	\label{XU}
\end{eqnarray}

System \eqref{XU} has three equilibria: $\mathbf 0=(0,0)$, $(m,0)$ and $(1,0)$, which correspond to the restrictions of $\mb{p}_0$, $\mb{p}_m$ and $\mb{p}_1$, respectively, to $\Pi_U$. 
	The equilibrium $(m,0)$ is a hyperbolic repeller and
$(1,0)$ is a hyperbolic saddle of \eqref{XU}. This result is a direct consequence of Hartman-Grobman theorem. Indeed, the linear part of \eqref{XU} is given by
	\begin{align*}
		DX_U(U,W)=\begin{pmatrix}
		0 && 1
		\\
		\dfrac{s}{d}\left(4U^3-3(m+1)U^2+2mU\right) && \dfrac{c}{d}
		\end{pmatrix}.
	\end{align*}
	Therefore, evaluation of $DX_U$ at $(m,0)$ and $(1,0)$ is given, respectively, by
	\begin{align*}
		DX_U(m,0)=\begin{pmatrix}
		0 && 1
		\\
		\dfrac{s}{d}m^2(m-1) && \dfrac{c}{d}		
		\end{pmatrix}, \qquad DX_U(1,0)=\begin{pmatrix}
		0 && 1
		\\
		\dfrac{s}{d}(1-m) && \dfrac{c}{d}
		\end{pmatrix}.
	\end{align*}
	Furthermore, the eigenvalues of $DX_U(m,0)$ and $DX_U(1,0)$ are given, respectively, {by
	\begin{gather*}
		\lambda_{\pm}^m=\dfrac{c\pm \sqrt{c^2-4sdm^2(1-m)}}{2d}, \qquad \lambda_\pm^1=\dfrac{c\pm \sqrt{c^2+4sd(1-m)}}{2d}.
	\end{gather*}
	}
	Since $0<m<1$, then we have
	$	c^2-4sdm^2(1-m)<c^2.$
	Therefore, ${\rm Re}(\lambda_\pm ^m)>0$ and, hence, this proves that $(m,0)$ is a repeller. On the other hand, note that $c^2+4sd(1-m)>c^2$, so $\lambda_-^1<0<\lambda_+^1$.
	Thus, $(1,0)$ is a saddle.

In the case of the origin $(0,0)$ of \eqref{XU}, we have
	\begin{align*}
		DX_U(\mathbf 0)=\begin{pmatrix}
		0 & 1
		\\
		0 & \dfrac{c}{d}
		\end{pmatrix},
	\end{align*}
with associated eigenvalues $\lambda_1^0=0,$ {and $\lambda_2^0=c/d>0.$}
Hence, $\mathbf 0$ is a non-hyperbolic equilibrium.
	The eigenvectors of $DX_U(\mathbf 0)$ associated with $\lambda_1^0$ and $\lambda_2^0$ are $v_1^0=(1,0)^T$ and $v_2^0=(c,d)^T$, respectively. According to the centre manifold theorem~\cite{guckenheimer}, the origin of \eqref{XU} has a one-dimensional local centre manifold $W_\text{loc}^c(\mathbf 0)$ which is tangent to $v_1^0$ at $\mathbf 0$. This implies that $W_\text{loc}^c(\mathbf 0)$ can be represented locally as the graph of a function $W=W(U)$ that satisfies $W(0)=W'(0)=0$ (for further details, see~\cite{kuznetsov}). By taking the Taylor series expansion of this function around $U=0$, we have
$$
	W(U)=\sum_{k=2}^r a_kU^k+\mathcal O(U^{r+1}),
$$
where the coefficients $a_k\in \mathbb{R}$, and $\mathcal O(U^{r+1})$ are terms of order $r+1$ and higher of the Taylor series of $W(U)$. Here, the  coefficients $a_k$ are determined by substitution of $W=W(U)$ into \eqref{XU}. Thus, after some calculations, we obtain	
$$
		W(U)=-\dfrac{ms}{c}U^2+\mathcal O(U^3). \label{W(U)}
$$
If we restrict \eqref{XU} to $W_\text{loc}^c(\mathbf 0)$, we obtain the scalar differential equation
	\begin{align*}
		U'=W(U)=-\dfrac{ms}{c}U^2+\mathcal O(U^3).
	\end{align*}
	Then, for $U>0$ small enough, we have that $U'<0$. Therefore, $\mathbf 0$ is a local attractor in $W_\text{loc}^c(\mathbf 0)$. Since $\lambda_2^0>0$, it follows that $\mathbf 0$ is a non-hyperbolic saddle point of \eqref{XU}.

	\begin{figure}
	\centering
	\includegraphics[trim = 8mm 0mm 10mm 0mm, clip,scale=0.8]{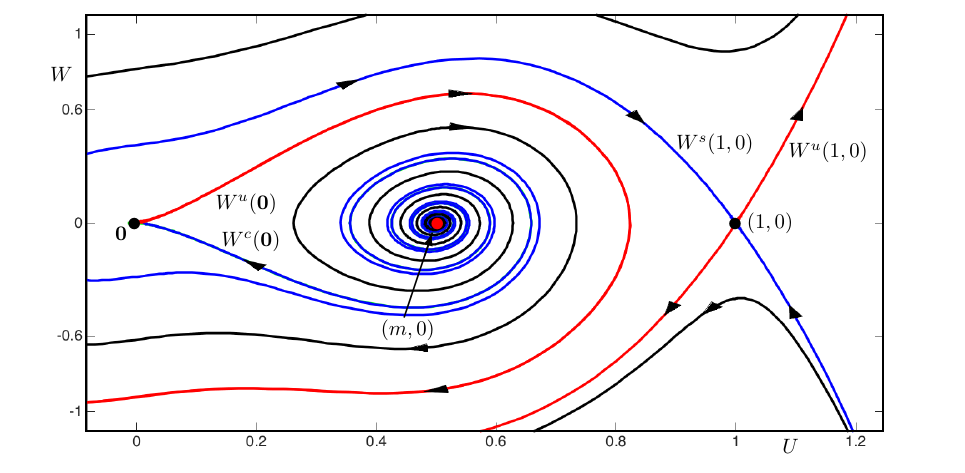}
	\caption{The centre manifold $W^c(\mathbf 0)$ extends itself (for $z<0$) along the flow of \eqref{XU} and forms a heteroclinic connection to $(m,0)$. This corresponds to a wave front in \eqref{edp}. Parameter values in \eqref{XU} are $(c,d,m,s)=(1,2.4,0.5,100)$.}
	\label{fig:VR0completo}
\end{figure}

Fig.~\ref{fig:VR0completo} shows the phase portrait of \eqref{XU}. The global centre manifold $W^c(\mathbf 0)$ extends itself (for $z<0$) along the flow of \eqref{XU} and forms a heteroclinic connection to $(m,0)$. This corresponds to a wave front in \eqref{edp} in which  $v\equiv0$ and $u\rightarrow0$ as $t\rightarrow\infty$. Moreover, since $W^u(m,0)$ is a two-dimensional unstable manifold, the intersection along $W^c(\mathbf 0)\cap W^u(m,0)$ is transversal. It follows that the associated wave front persists under small parameter variations. Notice that, under suitable parameter perturbations, the
manifolds $W^u(\mathbf 0)$ and $W^s(1,0)$ may come to intersect along a second heteroclinic orbit from $(0,0)$ to $(1,0)$.  Since this manifold intersection is non-transversal, the resulting heteroclinic orbit is structurally unstable, i.e., it might be broken under small parameter variation. However, the front itself could still persist in the PDE~\eqref{edp}, but perhaps with a slightly adjusted speed $c$.


Finally, a similar analysis of \eqref{sist} restricted to the invariant plane $\Pi_V$ reveals that there are {either no homoclinic, heteroclinic orbits nor} limit cycles with non-negative $V$-coordinates in $\Pi_V$. 

\section{Discussion}
\label{sec:discussion}

We have investigated a diffusive process in a qualitative model~\eqref{difusion} inspired by some  minimal mathematical ingredients of a rescaled predator-prey interaction from~\cite{aguirre}.
Our goal was to undertake a systematic identification of traveling waves that may be encountered when the kinetic terms ---based on those in~\cite{aguirre}--- 
are stated in their most elementary form. In particular, the analysis was performed searching for traveling wave solutions with strictly non-negative $(u,v)$ coordinates. 
 While the considered system may rise some questions regarding its interpretability, this ``prototype model" may be understood in the same spirit as investigating a sort of topological normal form: namely, if the kinetic terms are presented in their most reduced (and yet mathematically meaningful) form, what can one expect in terms of spatiotemporal behavior if diffusion is taken care of.  Nevertheless, upon taking these observations into consideration, we opted not to call variables $u$ and $v$ as ``prey" and ``predator", respectively, in order to avoid misleading conclusions and interpretations of our results.

 In~\eqref{difusion} the three main types of traveling waves we were looking for appear: wave pulses, wave fronts, and wave trains. Table \ref{tab:resumenOV} shows a list of all the traveling waves exposed throughout this paper, together with the regions of the bifurcation diagram in Fig.~\ref{fig:bifurcacion} where they can be found. We highlight that we made this table according to the obtained information, and the mentioned regions may not be the only ones where the given traveling waves can be found. Regarding the diffusion process, we found the condition $d>1$ favors the presence of chaotic orbits in the traveling frame, i.e., the associated ODE system \eqref{sist} is more likely to be chaotic when $D_u$ is higher than $D_v$. 
Specifically, we can find chaos when $(m,d)$ lies in the neighborhood of the $h_p$, $h_q$ curves (homoclinic chaos) and the Neimark-Sacker curve $NS$ in the bifurcation diagram; see also the discussion in \S\ref{sec:homoclinic}.

\begin{table}
	\caption{Summary of the main traveling wave solutions.}
	\label{tab:resumenOV}       
	\begin{tabular}{lll}
		\hline\noalign{\smallskip}
		Wave type & Region in $(m,d)$ plane & Orbit in phase space  \\
		\noalign{\smallskip}\hline\noalign{\smallskip}
		Type A wave pulse & Bifurcation curve $h_q$ & Focus-focus homoclinic orbit $\Gamma_q$ \\
		Type B wave pulse & Bifurcation curve $h_p$ & Saddle-focus homoclinic orbit $\Gamma_p$ \\
		Type C wave pulse & Neighborhoods of $h_q$ & Focus-focus 2-homoclinic orbit $\Gamma_q^2$ \\
		Type D wave pulse & Neighborhoods of $h_p$ & Saddle-focus 2-homoclinic orbit $\Gamma_p^2$ \\
		Type E wave pulse & Neighborhoods of $h_p$ & Saddle-focus 4-homoclinic orbit $\Gamma_p^4$ \\
		Type F wave pulse & Neighborhoods of $h_q$ & Focus-focus 3-homoclinic orbit $\Gamma_q^3$ \\
		Type A wave front & Region I & Heteroclinic orbit from $\mb{q}$ to $\mb{p}$ $\Gamma_{q,p}$ \\
		Type B wave front & $h_q$ and Regions III, IV, V & Heteroclinic orbit from $\mb{p}$ to $\mb{q}$ $\Gamma_{p,q}$ \\
		Type A wave train & Regions II, III, IV, V and VI & Limit cycle, $\Gamma$ \\
		Type B wave train & Regions III and V & 2-turns limit cycle, $\Gamma^2$ \\
		Type C wave train & Region III & 4-turns limit cycle, $\Gamma^4$ \\
		Type D wave train & Region III & 8-turns limit cycle, $\Gamma^8$ \\
		\noalign{\smallskip}\hline
	\end{tabular}
\end{table}

It is relevant to remember that there are regions where two or more types of traveling waves exist simultaneously. For example, if $(m,d)$ is in the intersection between $h_q$ and $h_p$, then there are two homoclinic orbits, {$\Gamma_q$} and $\Gamma_p$. Furthermore, let us remember that $\Gamma^2$ is branched from $\Gamma$ when $(m,d)$ crosses the period-doubling curve $PD$ from region II to III. However, $\Gamma$ does not disappear after this event. In fact, after the successive period doubling bifurcations, a large number of wave trains may be found together in the phase portrait of system \eqref{sist}. Finally, we should recall that when $(m,d)\approx(0.0463358,1.3080156)\in h_q$, there is a family of heteroclinic orbits, which go from $\mb{p}$ to $\mb{q}$. These arise from a transverse intersection between $W^u(\mb{p})$ and $W^s(\mb{q})$, which indicates that these solutions are robust under small changes of parameter values. 

A relevant fact regarding the homoclinic orbit $\Gamma_q$ is that it presents Shilnikov focus-focus homoclinic chaos in a concrete model vector field. Indeed, most of the studies about this bifurcation have been carried out (so far) only theoretically~\cite{foco-foco,wiggins}. Moreover, as far as we know, this work represents the first example of the actual computation of a global two-dimensional invariant manifold involved in a focus-focus homoclinic bifurcation in $\mathbb{R}^4$.


We also performed a thorough search for traveling waves connecting stationary states with $V=0$ to invariant objects with $U,V>0$ in each parameter region in Fig.~\ref{fig:bifurcacion} (In the particular case of $p_m$, we chose $c=10$ in order to ensure its eigenvalues are real). This procedure involved the computation of 3D and 4D unstable manifolds of equilibria $p_1$ and $p_m$, respectively, following the scheme from~\cite{siads20}. Unfortunately, all the computed orbits have negative $(U,V)$ coordinates for some $z>0$ and, hence, 
 did not match our initial criterium of non-negativity. 

While the amplitude of the solutions we found may be rather small, it is important to emphasize that our results are qualitative; that is, our aim was not to present specifically quantitative features of the predator-prey interaction  on which the system is based, but to show what kind of dynamical behavior one may hope to find, either theoretically or numerically. 
That said, the amplitude of the oscillations varies in markedly different scales for the rescaled variables, and variations in $V$ are much more pronounced than those of~$U$ for pulses, fronts and trains.
This difference in amplitudes can be explained by paying attention to the following observations: (i) quantities $u$ and $v$ spatially spread at distinct rates. This rate difference provides that the low-diffusivity variable promotes a heterogeneous aggregation of both variables along the domain; in consequence, both traveling-wave profiles propagate in a spatially non-homogenous shape, regardless of $d<1$ or $d>1$; 
(ii)~the  ratio in the original variables $u/v$  prevails as both amplitudes are of the same order in any scenarios here considered; (iii)~upon integrating system~\eqref{edp} along finite spatial domain, and assuming uniform convergence in time, we obtain that solutions satisfy formulae:
\begin{subequations}\label{eq:profinteg}
\begin{gather}
		\displaystyle\int\limits_{-L}^{L}\left[su(u-m)(1-u)(u+v)-auv\right]\:\text{d}x=\dfrac{d\mathcal{B}_1}{dt}\,, \label{eq:profintegA}\\[-2ex]
		\displaystyle\int\limits_{-L}^{L}\left[buv-gv(u+v)\right]\:\text{d}x=\dfrac{d\mathcal{B}_2}{dt} \,, \label{eq:profintegB}
\end{gather}
\end{subequations}
where the total ``masses" are given by $\mathcal{B}_1(t)=\int_{-L}^{L}u\:\text{d}x$ and $\mathcal{B}_2(t)=\int_{-L}^{L}v\:\text{d}x$, regardless of whether the boundary conditions for $u_x$ and $v_x$ vanish or cancel out each other at $x=\pm L$. Moreover, once we define the weighted total  mass by $\mathcal{B}(t):=b\mathcal{B}_1(t)+a\mathcal{B}_2(t)$, from~\eqref{eq:profinteg}, we obtain formula
\begin{gather}\label{eq:identity}
	\int\limits_{-L}^{L}\left[bsu(u-m)(1-u)-agv\right](u+v)\:\text{d}x=\dfrac{d\mathcal{B}}{dt}\,.
\end{gather}

Upon taking into account $L\gg1$, we get traveling-wave solutions of~\eqref{edp}, which are characterized by having profiles that keep their shape over time. Now, we consider a moving interval-frame $\mathcal{J}=[-L-ct,L-ct]$, which ``runs alongside'' the traveling profiles with the same speed~$c\geq0$. In so doing, integral in the left-hand side term of~\eqref{eq:identity} over $\mathcal{J}$ is constant for all $t\geq0$. That is, since traveling-wave profiles displacement only changes their position in time, the weighted total mass traveling speed is therefore conserved, i.e. $d\mathcal{B}/dt\equiv\mathcal{C}_0$, where $\mathcal{C}_0$ is constant.
Thus,  as $u+v\geq0$ for all $|x|\leq L$, the total  masses of the two  variables follow a conservation-like property for traveling-wave solutions.  Namely, as the wave variable $z=x+ct$ corresponds to a spatial translation for $t\geq0$, variable amplitudes balance each other to satisfy identity~\eqref{eq:identity} for a constant weighted total  mass rate of change $\mathcal{C}_0$.

It is interesting to note that identities in~\eqref{eq:profinteg} are also satisfied for stationary solutions; that is, upon setting $u_t=v_t=0$, we have relation~\eqref{eq:identity} with  $d\mathcal{B}/dt\equiv0$, when homogeneous Neumann boundary conditions are in place. Similar identities are obtained for Dirichlet, periodic and mixed-type boundary conditions, as well. In addition, this model may be able to develop diffusion-driven instabilities by means of Turing bifurcations. Nonetheless, in order to perform such an analysis, a stationary spatial pattern setting must be taken into account.

Finally, we studied the stability of every wave solution we detected, and we found that each of them is unstable. We followed the approach from~\cite{brena2014,sandstede2002stability,sandstede2001stability}: In a general PDE of the form
\begin{align*}
	\mathbf u_t&= \mathbf f(\mathbf u) + D \, \mathbf u_{xx},
\end{align*}
traveling waves are functions $\mathbf U(z)$, where $z=x+ct$. In particular, defining the new pair of variables $(\xi,t)=(x+ct,t)$, we have that
\begin{align*}
	\mathbf u_t= c \, \mathbf u_\xi + \mathbf u_t, \qquad \mathbf u_{xx}= \mathbf u_{\xi \xi}.
\end{align*}
This implies that the system can be written as
\begin{equation}\label{eq:stability}
	\mathbf u_t = -c \, \mathbf u_\xi + \mathbf f(\mathbf u) + D \, \mathbf u_{\xi \xi}.
\end{equation}
In our case, $\mathbf u=(u,v)$, $\mathbf f$ stands for the reaction part of  \eqref{edp}, and $D=
\left(
\begin{array}{cc}
d & 0\\
0&1
\end{array}
\right).
$
Notice that when $\mathbf u_t=\mathbf 0$ in \eqref{eq:stability} we recover the ODE system \eqref{edp2} for traveling waves. This means that traveling waves are a stationary solution of \eqref{eq:stability}. We then linearized it and evaluated the traveling waves we found. In particular, we obtained the spectral stability by numerically computing the ten eigenvalues with the largest real part of the Jacobian matrix of \eqref{eq:stability} and noticed that (in every run) the spectrum obtained had at least one eigenvalue with positive real part.
We also tried with different mesh refinements and the results were consistent.

 Last but not least, as a consequence of recasting a reaction-diffusion equation as a dynamical system with infinite dimensions, the multiplicity of solutions is large ---one for each initial condition. As the solutions we found turned out to be unstable, we focused on highlighting the zoo of travelling waves of the system stressing the analysis of the chaotic behavior of such solutions. Hence, no integration of the reaction-diffusion system was carried out since that would not help to clarify the ideas about the analysis here presented.
Disregarding the fact that the solutions are unstable, the analysis performed here is novel in the area of dynamical systems and we expect this to provide future guidance to study Shilnikov homoclinic focus-focus chaos with an emphasis on the stable and unstable manifolds of related homogeneous steady states.

\bibliographystyle{plain}
\bibliography{biblio_ago2224}

\end{document}